\documentclass[aps,pra,reprint,twocolumn,floatfix]{revtex4-2}
\usepackage{float}
\usepackage{orcidlink}
\usetikzlibrary{positioning}
\usepackage{algpseudocode}
\usepackage[linesnumbered,ruled,vlined]{algorithm2e}
\usepackage{amsmath,amssymb}
\usepackage{amsfonts}
\usepackage{graphicx} 
\usepackage[a4paper, top=0.75in, bottom=1in, left=0.75in, right=0.75in]{geometry} 
\usepackage{comment}
\usepackage{tikz}
\usetikzlibrary{positioning}
\usepackage{dsfont}
\usepackage[edges]{forest}
\usepackage{amsmath}
\usepackage{adjustbox}
\usepackage{enumitem}
\usepackage{mathrsfs}
\usepackage{graphicx}
\usepackage{braket}
\usepackage{amsthm}

\usepackage{tikz}
\usepackage{quantikz}
\usepackage{dsfont}
\usepackage{eso-pic}
\usetikzlibrary{arrows.meta,calc,positioning}
\usepackage{hyperref}

\newtheorem{theorem}{Theorem}[section] 
\newtheorem{proposition}[theorem]{Proposition}
\newtheorem{corollary}[theorem]{Corollary}

\theoremstyle{definition}
\newtheorem*{definition}{Definition}

\theoremstyle{remark}
\newtheorem*{remark}{Remark}

\newcommand{\kagwavecopyright}{%
  \AddToShipoutPictureFG*{%
    \AtPageLowerLeft{%
      \hspace*{0.5in}%
      \raisebox{0.25in}{%
        \scriptsize*\textcopyright\ KAGWAVE.%
      }%
    }%
  }%
}

\extrafloats{100}
\begin{document}

\kagwavecopyright

\title{The Quantum Character Transform}
\author{
Kagwe A. Muchane\orcidlink{0009-0008-6649-3975}*%
}
\affiliation{
Department of Computer Science, North Carolina State University, Raleigh, NC 27606, USA \\
Department of Mathematics, North Carolina State University, Raleigh, NC 27606, USA
}

\begin{abstract}
Bastidas et al. (2025) highlighted the importance of exploiting group symmetries to reduce \mbox{computational} overhead and unify simulation methods across both condensed matter physics and ab initio quantum chemistry. I discuss the discrete character transform, which is the group Fourier transform restricted to $ Z(\mathbb{C}[\mathcal G])$, the center of the group algebra, where class functions form a commutative convolution algebra and irreducible characters are its spectral basis. The normalized character \mbox{table} of a non-abelian group can be obtained through a restriction operator, which means an efficient \mbox{quantum} character transform for any finite group with an efficient quantum Fourier transform. I give complementary quantum algorithms for class convolution based on direct class-sum block \mbox{encoding} and spectral implementation in the character basis, explaining their applications to symmetry-adapted many-body dynamics, molecular and spectroscopic symmetry analysis, and quantum walks, as well as characterizing the central information used in hidden-subgroup problems, including its limitations for graph isomorphism.
\end{abstract}

\maketitle

\pagestyle{plain}

\section{Introduction}

Quantum chemistry is widely considered one of the most promising and impactful commercial applications for early quantum computers, because molecular electronic structure is in essence quantum mechanical and its accurate classical description can require resources that grow prohibitively with system size \cite{mcardle2020quantum, cao2019quantum}. Yet, representing a molecular Hamiltonian on a quantum computer does not by itself make a calculation efficient. Practical algorithms must also identify and exploit the symmetries that constrain the physically relevant state space. Point-group, permutation, spin, and other finite symmetries decompose molecular Hilbert spaces into invariant sectors, where states and observables can be organized by irreducible representations. Character tables describe exactly this decomposition, supporting the classification of molecular orbitals and normal modes, the derivation of infrared and Raman selection rules, and the block diagonalization of symmetry-preserving Hamiltonians \cite{hamermesh1962group,cotton1990chemical,birman1974theory,Bunker1979}. Their pure mathematical foundations reach back to Frobenius’s analysis of the group determinant \cite{frobenius1896gruppencharaktere,frobenius1897gruppendeterminante}, Burnside’s systematization of finite-group representation theory \cite{burnside1897theory,burnside1911theory}, and Schur’s development of linear and projective representation theory \cite{schur1904darstellung,schur1905neue}. As chemistry becomes a major target for quantum hardware, computing and manipulating character data has become more than a classificatory concern.

The abelian quantum Fourier transform and Kitaev’s phase-estimation algorithm showed that spectral information can be extracted efficiently under a change of basis. Beyond the abelian case, efficient transforms over specific nonabelian families, most prominently Beals’s transform for the symmetric group \cite{Beals1997} and the Schur transform of Bacon, Chuang, and Harrow \cite{BaconChuangHarrow2006}, proved that noncommutative harmonic analysis can also be implemented efficiently when representation bases and branching rules allow recursive quantum circuits. In the literature, the term "character transform" has described two distinct operations. Greg Kuperberg’s early use of the term \cite{Kuperberg2005} should not be confused with the class-condensed unitary operator used today. While related, in his 2003 work on the Hidden Subgroup Problem, Kuperberg uses the term to refer to a non-abelian Fourier transform on the full group representation space. Over a decade prior, a more compact, class-based change of basis was already implicit in character-theoretic literature used throughout classical chemistry to analyze molecular systems by symmetry. In Cotton's book \textit{Chemical Applications of Group Theory}, the Great Orthogonality Theorem gives the class-weighted orthogonality relations underlying molecular character tables \cite{cotton1990chemical}. With normalized conjugacy-class coordinates, these relations make the matrix with entries
\[
U_{\Gamma, C}
=
\sqrt{\frac{|C|}{|\mathcal G|}}\,
\chi_\Gamma(C)
\]
unitary, with rows indexed by irreducible characters and columns by conjugacy classes. In 2015, Bermejo-Vega and Zatloukal placed this class-to-character change of basis within quantum Fourier analysis over abelian hypergroups, showing how the group quantum Fourier transform can access the character basis of conjugation-invariant states \cite{BermejoVegaZatloukal2015}. Most recently, Bastidas, Liu, Chuang, and colleagues elevated this class-based unitary to a definitive quantum operation \cite{bastidas2025}. Recognizing that element-level tracking can become resource-intensive on near-term quantum hardware, Bastidas et al.\ deployed this class-based unitary directly inside quantum circuits as an active symmetry transform. By shifting the state vector into conjugacy-class and irreducible-representation coordinates, they enable symmetry-adapted projection and block diagonalization, such that physically relevant symmetry sectors can be isolated and manipulated coherently; they further demonstrated this symmetry-adapted subroutine for molecular hydrogen, $\mathrm{H}_2$, experimentally on IBM quantum hardware in the presence of noise.

That transform is the \textit{Quantum Character Transform}. For a finite group \(\mathcal G\), the class sums form a basis of the center \(Z(\mathbb C[\mathcal G])\), while the irreducible group characters provide its spectral basis. The class algebra is the commutative convolution algebra of conjugacy-invariant functions on \(\mathcal G\), and the normalized character table gives the unitary change of basis between normalized conjugacy-class and irreducible-character coordinates. The discrete character transform is non-abelian Fourier analysis restricted to the center of the group algebra, and because commutativity is often easier to work with, it enables efficient spectral calculus for \(Z(\mathbb C[\mathcal G])\). This has an important algorithmic consequence, as once the QCT is available, one can pass in both directions between central operators and their spectra. Given a class function \(f\), the transform exposes the spectrum of a class convolution operator \(T_f\); conversely, a prescribed eigenvalue profile over irreducible sectors creates, through the inverse QCT, an associated class operator. This operator synthesis approach situates the quantum character transform in the same conceptual family as modern quantum singular-value and polynomial transformation methods, where diagonal manipulations in the character basis become operator transformations in the original basis.

Seen through its history, the expected applications of the discrete character transform are wide. If $f$ is supported on a conjugacy class or union of classes, then $T_f$ is the adjacency operator of a normal Cayley graph, so the character transform allows direct access to its spectrum and to the quantum-walk Hamiltonians generated by such graphs. If $T_f$ is used inside a phase-estimation-style routine, irreducible sectors with distinguishable eigenvalues can be discriminated, which connects directly to symmetry projection and to the generalized symmetry-adapted constructions advocated by Bastidas et al. If one instead applies polynomial or rational functions of $T_f$, the same mechanism becomes a central spectral filter, a Green’s-function/resolvent primitive, or a tool for isolating selected symmetry sectors before simulation or measurement.  The character transform is also most compelling from the standpoint of compression. The full nonabelian QFT acts on \(\mathbb{C}[\mathcal G]\), whose dimension is $|\mathcal G|$, whereas the QCT acts only on the class algebra, whose dimension is $d = k(\mathcal G), $ the number of conjugacy classes. The natural state-space dimension of the transform is $k(\mathcal G) \leq |\mathcal G|$,  and for binary encoding, $\lceil \log_2 k(\mathcal G) \rceil$ qubits. For noncommutative groups such as $S_n$, this becomes a dramatic asymptotic compression of the representation space. That compression is capitalizable whenever the relevant data is already conjugacy-invariant, and it has a useful connection to symmetry-aware machine learning. Class-algebra features are the compressed harmonic representatives of data defined only up to relabeling or conjugation, much as invariant and equivariant architectures exploit group structure to reduce model complexity and improve generalization. 
 
I develop an algorithmic theory of the center \(Z(\mathbb C[\mathcal G])\). We first establish the discrete character transform's correct relation to the full non-abelian quantum Fourier transform, showing that the normalized character table of $\mathcal G$ is procured by restricting the group Fourier transform to the conjugation-invariant subspace and coherently removing the redundant matrix-index degrees of freedom. This general algorithm is efficient for all known finite groups admitting an efficient quantum Fourier transform, and unites earlier character-basis works in a common and comprehensive operator-theoretic framework. I then use modern group convolution methods to build quantum algorithms specifically for class convolution, and show how they support central-operator synthesis and spectral manipulation. I also apply this transform  to symmetry-adapted many-body and molecular problems, quantum walks, and hidden-subgroup settings, where the information retained by restriction to the center and the limit of its usefulness for graph-isomorphism-type problems is clarified. An open-source implementation and benchmarking framework for the QCT and the class-convolution algorithms is under active development. Preliminary quantum-hardware experiments have been conducted for the QCT over $D_n$ and class convolution over $S_4$, with scaling results and empirical comparisons against classical GPU implementations to be reported in a subsequent version of this manuscript.
\section{Group Algebras}

The group algebra $\mathbb{C}[\mathcal{G}]$ is the complex vector space spanned by the elements of $\mathcal{G}$, consisting of all formal sums
\[
x=\sum_{g\in\mathcal G}f(g)\,g,
\qquad
f:\mathcal G\to\mathbb C.
\]
Multiplication is defined by bilinearly extending the group law, making $\mathbb{C}[\mathcal{G}]$ a generally noncommutative group ring of
$\mathcal G$ over $\mathbb C$. The function algebra $\mathbb{C}(\mathcal{G})$ is the space of
complex-valued functions on $\mathcal G$, equipped with pointwise multiplication, $(f_1f_2)(g)=f_1(g)f_2(g)$. Every group algebra element $x$ is uniquely determined by its coefficient function
$f:\mathcal G\to\mathbb C$, so there is a canonical vector space isomorphism $\mathbb C[\mathcal G]\cong\mathbb C(\mathcal G)$. Although isomorphic as vector spaces, they differ as algebras, $\mathbb C(\mathcal G)$ uses pointwise multiplication, whereas under this identification, the product of two group algebra elements
$x_1=\sum_g f_1(g)g$ and
$x_2=\sum_g f_2(g)g$
is represented by the convolution
\begin{equation}
(f_1*f_2)(h)
=
\sum_{g\in\mathcal G}
f_1(g)f_2(g^{-1}h),
\label{eq-conv}
\end{equation}
which forms the algebraic foundation for Fourier analysis on $\mathcal G$.
The Fourier transform diagonalizes convolution, reducing convolution to pointwise multiplication in the Fourier basis.

The center of the group algebra, $Z(\mathbb C[\mathcal G])$, consists of the elements $x$ commuting with every other element. Since the elements $g \in \mathcal G$ form a basis of $\mathbb C[\mathcal G]$, it suffices to verify commutation with each basis element. Under the canonical vector space isomorphism
$\mathbb C[\mathcal G]\cong\mathbb C(\mathcal G)$,
this condition is the same as the coefficient function being constant on conjugacy classes, that is, to being a class function. The class algebra, $\mathrm{Cl}(\mathcal G)$, is the subalgebra of $\mathbb C[\mathcal G]$ consisting of all class functions,
\begin{equation}
\mathrm{Cl}(\mathcal G)
:=
\left\{
f\in\mathbb C(\mathcal G)
\mid
f(hgh^{-1})=f(g),
\ \forall g,h\in\mathcal G
\right\}.  
\label{eq-class-alg}
\end{equation}

As a vector space, $\mathrm{Cl}(\mathcal G)$ is indexed by the conjugacy classes of $\mathcal G$. For each conjugacy class $C \in\operatorname{Conj}(\mathcal G)$, define its indicator function by
$\mathds{1}_{C}(g)$. The collection $\{\mathds{1}_C\}$ forms a basis for $\mathrm{Cl}(\mathcal G)$. To obtain an orthonormal basis, these indicator functions must be normalized. The normalized counting inner product on $\mathbb C(\mathcal G)$ is
$
\langle f_1,f_2\rangle
=
\frac{1}{|\mathcal G|}
\sum_{g\in\mathcal G}
f_1(g)\overline{f_2(g)},
$
and following,
$
\langle\mathds{1}_C,\mathds{1}_C\rangle
=
\frac{|C|}{|\mathcal G|}.
$
The size of a conjugacy class is related to the size of its centralizer $Z_g=\{h\in\mathcal G\mid hg=gh\}$
by the orbit--stabilizer theorem,
$
|C|
=
\frac{|\mathcal G|}{|Z_g|}.
$
Hence the normalized class basis is
\begin{equation}
\{\hat{\mathds{1}}_C(g)\}
=
\begin{cases}
\sqrt{|Z_g|}, & g\in C,\\[1ex]
0, & \text{otherwise},
\end{cases}
\label{class-basis}
\end{equation}
where $g$ is any representative of the conjugacy class $C$. The class algebra inherits the inner product
$
\langle f_1,f_2\rangle
=
\sum_C
\frac{|C|}{|\mathcal G|}
f_1(C)\overline{f_2(C)},
$
under which the irreducible characters form an orthonormal basis.

\subsection{The Quantum Character Transform}

The irreducible characters are a special family of class functions, namely, the Fourier basis for the center of the group algebra, satisfying
\[
\chi_\Gamma(g)=\chi_\Gamma(hgh^{-1}),
\qquad
\forall\,g,h\in\mathcal G.
\]
Together with the normalized conjugacy-class indicator functions, they form two orthonormal bases of the class algebra.

\begin{definition}
The \emph{Discrete Character Transform} is the change of basis between the conjugacy-class basis of \(Z(\mathbb{C}[\mathcal{G}])\) and the irreducible character basis,
\[
U_{\mathrm{CT}}:
Z(\mathbb C[\mathcal G])
\longrightarrow
\mathbb C[\widehat{\mathcal G}].
\]
\end{definition}

For each conjugacy class \(C \in \operatorname{Conj}(\mathcal G)\), define its class sum
$K=\sum_{g\in C}g\in Z(\mathbb C[\mathcal G]) $,
these class sums form a basis of the center \(Z(\mathbb C[\mathcal G])\). Under amplitude encoding, the normalized indicator functions become the class states
$
|C\rangle
=
|C|^{-1/2}
\sum_{g\in C}|g\rangle.
$
Likewise, each irreducible character defines the normalized quantum state
$
|\chi_\Gamma\rangle
=
\sum_C
\sqrt{|C|/|\mathcal G|}\,
\chi_\Gamma(C)|C\rangle.
$ 
From this unitary representation, comes the Quantum Character Transform (QCT),
\begin{equation}
\sum_C a_C|C\rangle
\longmapsto
\sum_{\Gamma\in\widehat{\mathcal G}}
\left(
\sum_C
a_C\,
\chi_\Gamma(C)^*
\sqrt{\frac{|C|}{|\mathcal G|}}
\right)
|\Gamma\rangle.
\end{equation}

Bastidas et al.\ introduced the Quantum Character Transform as the normalized character table appearing in a symmetry-adapted projection circuit. Their framework indexes the input register by conjugacy classes and the output register by irreducible representations, thereby obtaining the normalized character matrix, whose
orientation as a forward or inverse transform determines whether
its entries are
$
\sqrt{|C|/|\mathcal G|}\,\chi_\Gamma(C)^*
$
or their complex conjugates.
This operational definition established the transform as a useful ingredient for coherent projection onto symmetry sectors, but left implicit the algebraic space on which it acts and its precise relation to the non-abelian group Fourier transform. Additionally, they extract the character transform from the full group Fourier through a probabilistic postselection of the matrix-index registers, which introduces a success-probability overhead. The following theorem shows that this projection is unnecessary; the redundant matrix-index degrees of freedom can instead be removed coherently and deterministically.

\begin{theorem}[Center restriction]
\label{thm:center-restriction}
The Quantum Character Transform is the restriction of the group Fourier transform to the center, 
\[
\left.
\mathscr F_{\mathcal G}
\right|_{Z(\mathbb C[\mathcal G])}
\cong
U_{\mathrm{QCT}}.
\]
\end{theorem}

\begin{proof}
Let \(\mathcal G\) be a finite group,
\(\widehat{\mathcal G}\) a complete set of inequivalent irreducible unitary representations, and
$
\mathscr F_{\mathcal G}:
\mathbb C[\mathcal G]
\longrightarrow
\bigoplus_{\Gamma\in\widehat{\mathcal G}}
V_\Gamma\otimes V_\Gamma^*
$ the nonabelian Fourier transform.
Applying the Fourier transform to a normalized conjugacy-class state,
\[
\begin{aligned}
\mathscr F_{\mathcal G}|C\rangle
&=
\frac{1}{\sqrt{|C|}}
\sum_{g\in C}\mathscr F_{\mathcal G}|g\rangle  \\
&=
\sum_{\Gamma\in\widehat{\mathcal G}}
\sqrt{\frac{d_\Gamma}{|\mathcal G||C|}}
\sum_{i,j}
\left(
\sum_{g\in C}\rho_\Gamma(g^{-1})
\right)_{ij}
|\Gamma,i,j\rangle .
\end{aligned}
\]

Since \(C^{-1}\) is a conjugacy class, the operator
\(\sum_{g\in C}\rho_\Gamma(g^{-1})\) commutes with every \(\rho_\Gamma(h)\). By Schur's lemma,
$
\sum_{g\in C}\rho_\Gamma(g^{-1})
=
\alpha_{\Gamma,C}I_{d_\Gamma}
$
for some scalar \(\alpha_{\Gamma,C}\). Taking traces,
$
\alpha_{\Gamma,C}
=
\frac{|C|\chi_\Gamma(C^{-1})}{d_\Gamma}
=
\frac{|C|\chi_\Gamma(C)^*}{d_\Gamma},
$
and therefore
\[
\sum_{g\in C}\rho_\Gamma(g^{-1})
=
\frac{|C|\chi_\Gamma(C)^*}{d_\Gamma}I_{d_\Gamma}.
\]
The identity matrix spans the unique invariant line in
\(V_\Gamma\otimes V_\Gamma^*\), namely
\(
\operatorname{vec}(I_{V_\Gamma})
=
\sum_i|i,i\rangle.
\)
Let
$
|I_\Gamma\rangle
=
\frac{1}{\sqrt{d_\Gamma}}
\operatorname{vec}(I_{V_\Gamma}).
$
Substituting,
\[
\begin{aligned}
\mathscr F_{\mathcal G}|C\rangle
&=
\sum_{\Gamma\in\widehat{\mathcal G}}
\sqrt{\frac{|C|}{|\mathcal G|d_\Gamma}}\,
\chi_\Gamma(C)^*
\sum_{i=1}^{d_\Gamma}|\Gamma,i,i\rangle\\
&=
\sum_{\Gamma\in\widehat{\mathcal G}}
\sqrt{\frac{|C|}{|\mathcal G|}}\,
\chi_\Gamma(C)^*
|\Gamma\rangle|I_\Gamma\rangle.
\end{aligned}
\]

Thus every conjugacy-class state is mapped into the invariant line of each irreducible block, and
$
\operatorname{Im}
\!\left(
\left.
\mathscr F_{\mathcal G}
\right|_{Z(\mathbb C[\mathcal G])}
\right)
=
\bigoplus_{\Gamma\in\widehat{\mathcal G}}
\mathbb C\,|\Gamma\rangle|I_\Gamma\rangle.
$
Under the canonical unitary identification
\(
|\Gamma\rangle|I_\Gamma\rangle
\mapsto
|\Gamma\rangle,
\)
the restriction acts by
\[
|C\rangle
\longmapsto
\sum_{\Gamma\in\widehat{\mathcal G}}
\sqrt{\frac{|C|}{|\mathcal G|}}\,
\chi_\Gamma(C)^*\,
|\Gamma\rangle,
\]
whose matrix is the normalized character table.
\end{proof}

\begin{remark}
The state
$
|I_\Gamma\rangle
=
d_\Gamma^{-1/2}\operatorname{vec}(I_{V_\Gamma})
$
is the unique (up to phase) normalized \((\rho_\Gamma\otimes\rho_\Gamma^*)\)-invariant vector in \(V_\Gamma\otimes V_\Gamma^*\). Under the Choi--Jamiołkowski correspondence
\cite{jamiolkowski1972linear, choi1975completely}, it represents the identity intertwiner on \(V_\Gamma\), and is the canonical maximally entangled state associated with the irrep $V_\Gamma$. Restricting the Fourier transform to the center removes every noncentral matrix direction, leaving this canonical invariant vector in each irreducible block. The entries of the normalized character table are the expansion coefficients of conjugacy-class states in the basis of these canonical invariant vectors.
\end{remark}

This character-basis reduction appears in Ref.~\cite[Appendix~D.1]{BermejoVegaZatloukal2015}. The hypergroup algebra of the abelian conjugacy-class hypergroup \(\overline{\mathcal G}\) is isomorphic to \(Z(\mathbb C[\mathcal G])\), so our formulation positions the character transform inside the non-abelian Fourier transform as a canonical restriction to the central subspace.
\section{Class Convolution}
Let
\(
f_1, f_2\in \mathrm{Cl}(\mathcal G)
\)
be class functions. Left convolution by \(f_1\) defines the action of a convolution operator $T_{f_1}$ by 
$
T_{f_1} f_2
=
f_1*f_2.
$
When \(\mathcal G\) is abelian, \(Z(\mathbb C[\mathcal G])
=\mathbb C[\mathcal G]\), and \(T_f\) is the usual translation-invariant group-circulant convolution operator. For cyclic groups, this is the familiar
circulant Toeplitz operator. For a general finite group, \(T_f\) is its central, conjugacy-invariant counterpart. Since the convolution of two class functions is again a class function, \(T_f\) restricts to an operator on \(Z(\mathbb C[\mathcal G])\). Let \(C_1,\ldots,C_d\) be the conjugacy classes of \(\mathcal G\). In the
normalized class basis, the matrix elements of \(T_f\) are
\begin{equation}
(T_f)_{ij}
=
\langle C_i|T_f|C_j\rangle
=
\frac{1}{\sqrt{|C_i||C_j|}}
\sum_{g\in C_i}\sum_{h\in C_j}f(gh^{-1}).
\end{equation}
Choosing the alternative convolution convention replaces \(f(gh^{-1})\) by
\(f(g^{-1}h)\); the two are related by the involution
\(f^\vee(g)=f(g^{-1})\).
Class convolution operators commute because
\(Z(\mathbb C[\mathcal G])\) is commutative,
$
T_{f_1}T_{f_2}
=
T_{f_1*f_2}
=
T_{f_2*f_1}
=
T_{f_2}T_{f_1}.
$
Under the involution
$
f^*(g)=\overline{f(g^{-1})},
$
one has
$
T_f^\dagger=T_{f^*}.
$
Also,
$
T_fT_f^\dagger=T_{f*f^*}=T_{f^**f}=T_f^\dagger T_f,
$
so every class convolution operator is \textit{normal}. It is Hermitian when
$
f(g^{-1})=\overline{f(g)}
$
for all \(g\in\mathcal G\), and real symmetric when \(f\) is real-valued and invariant under inversion.

The normalized irreducible characters form the common eigenbasis of these operators. To determine the eigenvalue, associate to \(f\) the central group-algebra element
$
z_f^\vee=\sum_{g\in\mathcal G}f(g)g^{-1}.
$
\[
\begin{aligned}
\rho_\Gamma(z_f^\vee)
=
\widehat f(\Gamma)I_{V_\Gamma},
\quad
\widehat f(\Gamma)
&=
\frac{1}{d_\Gamma}
\sum_{g\in\mathcal G}
f(g)\chi_\Gamma(g^{-1}) \\
&=
\frac{1}{d_\Gamma}
\sum_{i=1}^d
|C_i|\,f(C_i)\,
\overline{\chi_\Gamma(C_i)}.
\label{eq-eigenval}
\end{aligned}
\]
It follows that
$
T_f|\chi_\Gamma\rangle
=
\widehat f(\Gamma)|\chi_\Gamma\rangle.
$
Writing
$
D_{\widehat f}=\operatorname{diag}(\widehat f(\Gamma))_{\Gamma\in\widehat{\mathcal G}},
$
we obtain the following.

\begin{corollary}
The Discrete Character Transform simultaneously diagonalizes the entire algebra of class convolution operators,
\begin{equation}
U_{\mathrm{CT}}^{}
T_f
U_{\mathrm{CT}}^{-1}
=
D_{\widehat f}.
\label{eq-diag}
\end{equation}
\end{corollary}
Rearranging,
$
T_f
=
U_{\mathrm{CT}}^{-1}
D_{\widehat f}
U_{\mathrm{CT}}^{}.
$
Thus the regular representation of the class convolution algebra is equivalent to the algebra of diagonal operators on the irreducible-character labels,
\[
\{T_f:f\in \mathrm{Cl}(\mathcal G)\}
\cong
\bigoplus_{\Gamma\in\widehat{\mathcal G}}\mathbb C.
\]

\noindent
\textbf{\textit{Class convolution over $S_4$.}} As a running example, we focus on $S_4$. Let the conjugacy classes $C_i$ be ordered as
$
C_1=(1^4),\;
C_2=(2,1^2),\;
C_3=(2^2),\;
C_4=(3,1),\;
C_5=(4),
$
with respective sizes
$
|C_1|=1,\; |C_2|=6,\; |C_3|=3,\;
|C_4|=8,\; |C_5|=6.
$
In this ordering, the normalized character table defines the
$S_4$ quantum character transform below.
\begin{table}[H]
\centering
\caption{Normalized character table of $S_4$.}
\label{tab:qct-s4}
\renewcommand{\arraystretch}{1.35}
\begin{tabular}{c|ccccc}
$\mathrm{QCT}_{S_4}$
& $(1^4)$ & $(2,1^2)$ & $(2^2)$ & $(3,1)$ & $(4)$ \\
\noalign{\vskip 3pt}
\hline
\noalign{\vskip 3pt}
$[4]$
& $\dfrac{1}{\sqrt{24}}$
& $\dfrac{1}{2}$
& $\dfrac{1}{2\sqrt{2}}$
& $\dfrac{1}{\sqrt{3}}$
& $\dfrac{1}{2}$ \\
$[1^4]$
& $\dfrac{1}{\sqrt{24}}$
& $-\dfrac{1}{2}$
& $\dfrac{1}{2\sqrt{2}}$
& $\dfrac{1}{\sqrt{3}}$
& $-\dfrac{1}{2}$ \\
$[3,1]$
& $\dfrac{3}{\sqrt{24}}$
& $\dfrac{1}{2}$
& $-\dfrac{1}{2\sqrt{2}}$
& $0$
& $-\dfrac{1}{2}$ \\
$[2,1,1]$
& $\dfrac{3}{\sqrt{24}}$
& $-\dfrac{1}{2}$
& $-\dfrac{1}{2\sqrt{2}}$
& $0$
& $\dfrac{1}{2}$ \\
$[2,2]$
& $\dfrac{2}{\sqrt{24}}$
& $0$
& $\dfrac{1}{\sqrt{2}}$
& $-\dfrac{1}{\sqrt{3}}$
& $0$
\end{tabular}
\end{table}

The rows are indexed by the irreducible representations of $S_4$, labeled by the partitions $[4]$, $[1^4]$, $[3,1]$, $[2,1,1]$, and $[2,2]$. For a class function $f$, write
$
f=\sum_{i=1}^{d} f(C_i)\mathds{1}_{C_i}.
$
Ergo,
\begin{equation}
T_f=\sum_{i=1}^{d} f(C_i)T_{C_i},
\qquad
T_{C_i}:=T_{\mathds{1}_{C_i}}.    
\end{equation}
The operators $T_{C_i}$ form a basis for the algebra of class convolution operators. In the normalized conjugacy-class basis, for $S_4$ these operators are
\begin{equation}
\begin{aligned}
T_{(1^4)} &= I_5,\\[-0.1em]
T_{(2,1^2)} &=
\begin{psmallmatrix}
0&\sqrt6&0&0&0\\
\sqrt6&0&\sqrt2&2\sqrt3&0\\
0&\sqrt2&0&0&2\sqrt2\\
0&2\sqrt3&0&0&2\sqrt3\\
0&0&2\sqrt2&2\sqrt3&0
\end{psmallmatrix},\\[0.2em]
T_{(2^2)} &=
\begin{psmallmatrix}
0&0&\sqrt3&0&0\\
0&1&0&0&2\\
\sqrt3&0&2&0&0\\
0&0&0&3&0\\
0&2&0&0&1
\end{psmallmatrix},\\[0.2em]
T_{(3,1)} &=
\begin{psmallmatrix}
0&0&0&2\sqrt2&0\\
0&4&0&0&4\\
0&0&0&2\sqrt6&0\\
2\sqrt2&0&2\sqrt6&4&0\\
0&4&0&0&4
\end{psmallmatrix},\\[0.2em]
T_{(4)} &=
\begin{psmallmatrix}
0&0&0&0&\sqrt6\\
0&0&2\sqrt2&2\sqrt3&0\\
0&2\sqrt2&0&0&\sqrt2\\
0&2\sqrt3&0&0&2\sqrt3\\
\sqrt6&0&\sqrt2&2\sqrt3&0
\end{psmallmatrix}.
\end{aligned}
\label{s4-classconv-matrices}
\end{equation}
These matrices commute and form the regular matrix representation of \(Z(\mathbb C[S_4])\). By \eqref{eq-diag}, the Character Transform simultaneously diagonalizes the basis operators,
$
U_{\mathrm{CT}}^{(S_4)}
T_{C_i}
\left(U_{\mathrm{CT}}^{(S_4)}\right)^\dagger
=
D_{C_i}
$,
and by the eigenvalue formula \eqref{eq-eigenval},
\[
\begin{aligned}
D_{(1^4)}
&=\operatorname{diag}(1,1,1,1,1),\\
D_{(2,1^2)}
&=\operatorname{diag}(6,-6,2,-2,0),\\
D_{(2^2)}
&=\operatorname{diag}(3,3,-1,-1,3),\\
D_{(3,1)}
&=\operatorname{diag}(8,8,0,0,-4),\\
D_{(4)}
&=\operatorname{diag}(6,-6,-2,2,0).
\end{aligned}
\]
Every class convolution operator
$
T_f
$
is diagonalized by the same transform where
\[
\begin{aligned}
D_{\widehat f}&=
\operatorname{diag}\!\left(
\widehat f([4]),
\widehat f([1^4]),
\widehat f([3,1]),
\widehat f([2,1,1]),
\widehat f([2,2])
\right)\\
&=
\operatorname{diag}\;
\begin{psmallmatrix}
f(C_1)+6f(C_2)+3f(C_3)+8f(C_4)+6f(C_5)\\
f(C_1)-6f(C_2)+3f(C_3)+8f(C_4)-6f(C_5)\\
f(C_1)+2f(C_2)-f(C_3)-2f(C_5)\\
f(C_1)-2f(C_2)-f(C_3)+2f(C_5)\\
f(C_1)+3f(C_3)-4f(C_4)
\end{psmallmatrix}.   
\end{aligned}
\]
Conversely, reconstructing from the character-domain spectrum,
\[
\begin{aligned}
T_f
&=
\left(U_{\mathrm{QCT}}^{(S_4)}\right)^\dagger
D_{\widehat f}
U_{\mathrm{QCT}}^{(S_4)}
\\[0.3em]
&=
\begin{pmatrix}
a&b&c&d&e\\
b&r&g&h&i\\
c&g&j&k&\ell\\
d&h&k&m&h\\
e&i&\ell&h&r
\end{pmatrix},
\end{aligned}
\]
\[
\begin{aligned}
a&=f(C_1),&
r&=f(C_1)+f(C_3)+4f(C_4),\\
b&=\sqrt6\,f(C_2),&
g&=\sqrt2\!\left(f(C_2)+2f(C_5)\right),\\
c&=\sqrt3\,f(C_3),&
h&=2\sqrt3\!\left(f(C_2)+f(C_5)\right),\\
d&=2\sqrt2\,f(C_4),&
i&=2f(C_3)+4f(C_4),\\
e&=\sqrt6\,f(C_5),&
j&=f(C_1)+2f(C_3),\\
k&=2\sqrt6\,f(C_4),&
\ell&=\sqrt2\!\left(2f(C_2)+f(C_5)\right),\\
&&m&=f(C_1)+3f(C_3)+4f(C_4).
\end{aligned}
\]
The reconstructed class operator is symmetric and exhibits a pronounced Hankel-like repeated-entry structure. In fact, the apparent \(5\times5\) density is highly redundant. Since the first row already determines the class-function coefficients,
\[
\begin{aligned}
f(C_1)&=a, &
f(C_2)&=\frac{b}{\sqrt6}, &
f(C_3)&=\frac{c}{\sqrt3},\\
f(C_4)&=\frac{d}{2\sqrt2}, &
f(C_5)&=\frac{e}{\sqrt6},
\end{aligned}
\]
every remaining entry is fixed by \(a,b,c,d,e\).
\[
\begin{aligned}
r&=a+\frac{c}{\sqrt3}+\sqrt2\,d,
&
g&=\frac{b+2e}{\sqrt3},
&
h&=\sqrt2\,(b+e),
\\
i&=\frac{2c}{\sqrt3}+\sqrt2\,d,
&
k&=\sqrt3\,d,
&
\ell&=\frac{2b+e}{\sqrt3}.
\end{aligned}
\]
The entire \(25\)-entry operator can be written in terms of its five-entry first row. The class-convolution operators form a five-dimensional complex subspace of the 15-dimensional complex vector space of symmetric \(5\times5\) matrices. Indeed, eliminating the class-function coefficients in favor of \(a,b,c,d,e\) shows
\[
T_f=
\begin{psmallmatrix}
a
&
b
&
c
&
d
&
e
\\[0.2em]
b
&
a+\frac{c}{\sqrt3}+\sqrt2\,d
&
\frac{b+2e}{\sqrt3}
&
\sqrt2\,(b+e)
&
\frac{2c}{\sqrt3}+\sqrt2\,d
\\[0.2em]
c
&
\frac{b+2e}{\sqrt3}
&
a+\frac{2c}{\sqrt3}
&
\sqrt3\,d
&
\frac{2b+e}{\sqrt3}
\\[0.2em]
d
&
\sqrt2\,(b+e)
&
\sqrt3\,d
&
a+\sqrt3\,c+\sqrt2\,d
&
\sqrt2\,(b+e)
\\[0.2em]
e
&
\frac{2c}{\sqrt3}+\sqrt2\,d
&
\frac{2b+e}{\sqrt3}
&
\sqrt2\,(b+e)
&
a+\frac{c}{\sqrt3}+\sqrt2\,d
\end{psmallmatrix}.
\]
The entry relations and Hankel-like patterns visible in this \(S_4\) example are from the constrained class algebra, as opposed to the translation-invariant abelian group algebras in which convolution operators take Toeplitz form. While this specific reflection pattern depends on the group, and the ordering of conjugacy classes, the strong global redundancy that it exhibits persists for general class convolution operators.

\subsection{Implementation Methods}

Quantum class convolution is now of interest as a symmetry-reduced form of group convolution, but the underlying class algebra also has a long history in classical representation theory and physical applications. For symmetric groups, class-sum products have been studied in connection with systems of identical particles~\cite{Katriel1998ClassSum}, while Harrison and Scott used the standard representation of the \(S_3\) class algebra to construct a charged-lepton mass matrix whose spectral structure is determined by \(S_3\) character data~\cite{HarrisonScott2003}. Earlier this decade, quantum algorithms for the more general problem of group convolution and cross-correlation were developed by Castelazo et al.~\cite{Castelazo2022GroupConvolution}. Naively, convolution is computed through the full nonabelian Fourier transform, where it becomes multiplication within irreducible Fourier blocks. For a class function, Schur's lemma makes the Fourier blocks scalar, so the
internal matrix indices are unnecessary for the convolution spectrum. The diagonalization above provides two complementary implementations specialized to this central case. One may construct \(T_f\) directly as a linear combination of class-sum operators, or materialize it spectrally by conjugating a diagonal operator with the character transform. For Hermitian \(T_f\), the same spectral representation computes the evolution \(e^{-itT_f}\).

\subsubsection{Direct Class-Sum Block Encoding}

A direct implementation is natural when the class function is specified explicitly in the conjugacy-class basis, particularly when only a small number of conjugacy classes have nonzero coefficients, as occurs in sparse class functions and in many weighted Cayley-graph, combinatorial-Laplacian, and symmetry-adapted averaging operators. One may directly realize the convolution through a coherent superposition of left-regular group actions, yielding a block encoding whose complexity depends on the required state preparation and controlled group actions. This does not require an implementation of the QCT.

Let \(L_g\) denote the left-regular action
$
(L_gf_2)(h)=f_2(g^{-1}h).
$
For the convolution convention adopted above,
$
T_{C_i}
=
\sum_{g\in C_i}L_g,
$
and
\[
T_{f}
=
\sum_{i=1}^d
f(C_i)
\sum_{g\in C_i}L_g
=
\sum_{g\in\mathcal G}f(g)L_g.
\]
Thus, a class convolution operator is a linear combination of unitary left-regular group actions with coefficients constant on conjugacy classes. Define the normalization
\[
\alpha_{f}
=
\sum_{i=1}^d
|C_i|\,|f(C_i)|
=
\sum_{g\in\mathcal G}|f(g)|.
\]
Assuming coherent access to the class weights and phases of $f$, efficient preparation of normalized conjugacy-class states, and efficient controlled group multiplication, the following operations can be implemented coherently:
\begin{equation}
\begin{gathered}
\operatorname{PREP}_{f}|0\rangle
=
\frac{1}{\sqrt{\alpha_{f}}}
\sum_{i=1}^d
\sqrt{|C_i|\,|f(C_i)|}\,
|i\rangle
\frac{1}{\sqrt{|C_i|}}
\sum_{g\in C_i}|g\rangle,
\\
\operatorname{SELECT}_{f}:
|i\rangle|g\rangle|\psi\rangle
\longmapsto
\frac{f(C_i)}{|f(C_i)|}\,
|i\rangle|g\rangle L_g|\psi\rangle .
\end{gathered}    
\end{equation}
The phase may be chosen arbitrarily whenever \(f(C_i)=0\). Then
$
U_{f}
=
\left(\operatorname{PREP}_{f}^{\dagger}\otimes I\right)
\operatorname{SELECT}_{f}
\left(\operatorname{PREP}_{f}\otimes I\right)
$
satisfies
$
(\langle 0|\otimes I) U_f (|0\rangle\otimes I)
=
\frac{T_{f}}{\alpha_{f}},
$
so \(U_{f}\) is an \((\alpha_{f},a,0)\)-block encoding of \(T_{f}\), where \(a\) is the number of ancilla qubits needed by the class and group-element registers ~\cite{Gilyen2019QSVT}.

\begin{figure}[H]
\centering
\begin{quantikz}[row sep=0.35cm,column sep=0.55cm]
\lstick{$\ket{0}$}
    & \gate{\mathrm{PREP}_f}
    & \gate[wires=2]{\mathrm{SELECT}_{f}}
    & \gate{\mathrm{PREP}_f^{\dagger}}
    & \meter{}
\\
\lstick{$\ket{\psi}$}
    & \qw
    & \qw
    & \qw
    & \qw
\end{quantikz}
\caption{Direct class-sum block encoding of a class convolution operator. The
preparation oracle $\mathrm{PREP}_f$ prepares a weighted superposition over
conjugacy classes and their elements, while $\mathrm{SELECT}_f$ coherently
applies the corresponding left-regular action with the class-dependent phase.
Uncomputing the ancilla yields a block encoding of $T_f/\alpha_f$.}
\label{fig:lcu_classsum}
\end{figure}
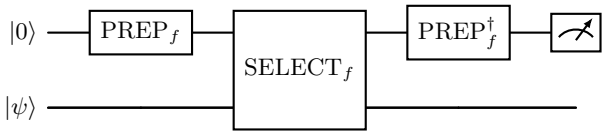
\begin{proposition}[Direct class-sum block encoding]
Let \(0\neq f\in \mathrm{Cl}(\mathcal G)\). \(T_{f}/\alpha_{f}\) has an exact block encoding with complexity
\[
2\operatorname{cost}(\operatorname{PREP}_{f})
+
\operatorname{cost}(\operatorname{SELECT}_{f}).
\]
\end{proposition}
Because each \(L_g\) acts on the full group algebra and \(T_f\) preserves its center, the same block encoding restricts to \(Z(\mathbb C[\mathcal G])\). This is the central specialization of the linear combination of unitaries method for general group convolution in Ref.~\cite{Castelazo2022GroupConvolution}. The block encoding is efficient whenever the required class-state preparation
and controlled left-regular action are also of low gate complexity.

Consider the conjugacy class \(C_{(3,1)}\) in \(S_4\),
\[
\{
(123),(132),(124),(142),
(134),(143),(234),(243)
\}.
\]
The associated class-convolution operator on the group-element basis is
$
T_{(3,1)}
=
\sum_{g\in C_{(3,1)}}L_g$
where
\;
$L_g\ket{x}=\ket{gx}.
$
\(T_{(3,1)}\) is an equal-weight linear combination of eight permutation unitaries, and because \(|C_{(3,1)}|=8\), a three-qubit selector register can be prepared using only Hadamard gates:
$
H^{\otimes 3}\ket{000}
=
\frac{1}{\sqrt8}
\sum_{r=0}^{7}\ket r.
$

\begin{figure}[H]
\centering
\begin{quantikz}
\lstick{\(\ket{000}\)}
    & \gate{H^{\otimes 3}}
    & \ctrl{1}
    & \gate{H^{\otimes 3}}
    & \meter{}
\\
\lstick{\(\ket{\psi}\)}
    & \qw
    & \gate{\operatorname{SELECT}_{(3,1)}}
    & \qw
    & \qw
\end{quantikz}
\caption{Direct class-sum block encoding of \(T_{(3,1)}\). The selector register is prepared in a uniform superposition over the eight elements of \(C_{(3,1)}\), and \(\operatorname{SELECT}_{(3,1)}\) coherently applies the corresponding left-regular actions. Uncomputing the selector yields \(T_{(3,1)}/8\) in the \(\ket{000}\) block.}
\label{fig:lcu}
\end{figure}
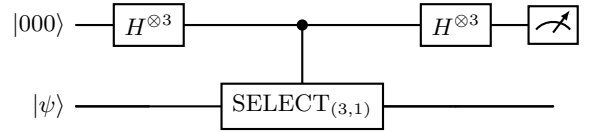

Choose an ordering \(C_{(3,1)}=\{g_0,\ldots,g_7\}\) and define
\[
\operatorname{SELECT}_{(3,1)}
=
\sum_{r=0}^{7}
\ket r\!\bra r\otimes L_{g_r}.
\]
The LCU circuit has upper-left block
$
\left(
\bra{000}H^{\otimes3}\otimes I
\right)
\operatorname{SELECT}_{(3,1)}
\left(
H^{\otimes3}\ket{000}\otimes I
\right)
=
\frac{1}{8}T_{(3,1)}.
$
Hence the circuit is an exact block encoding of the class sum with normalization factor
$
\alpha_{(3,1)}=8.
$
On the normalized conjugacy-class states
$
\ket{C_i},
$
the same operator has matrix entries
\[
\left(T_{(3,1)}\right)_{kj}
=
p_{(3,1),j}^{k}
\sqrt{\frac{|C_k|}{|C_j|}},
\]
where \(p_{(3,1),j}^{k}\) are the class-algebra structure constants (also intersection numbers) defined by
\[
K_{(3,1)}K_j
=
\sum_k p_{(3,1),j}^{k}K_k,
\qquad
K_j=\sum_{x\in C_j}x.
\]
For instance,
\[
T_{(3,1)}\ket{C_1}
=
\sum_{g\in C_{(3,1)}}\ket g
=
\sqrt8\,\ket{C_{(3,1)}}
=
2\sqrt2\,\ket{C_{(3,1)}},
\]
which is the entry \(2\sqrt2\) in the first column of the compressed class-basis matrix of \(T_{(3,1)}\).

Under the isomorphism \(S_4\cong T_d\), the class \(C_{(3,1)}\) is the eight proper threefold rotations of tetrahedral symmetry~\cite{cotton1990chemical}. Under a molecular representation \(\rho\), the class sum operator \(\sum_{g\in C_{(3,1)}}\rho(g)\) represents the equal-weight sum of the \(C_{3,1}\) symmetry operations, and applied to a molecular state, its normalized form acts as a symmetry-adapted filter according to the state's response under threefold rotational symmetry. In the character basis,
$
\frac{1}{8}
U_{\mathrm{QCT}}^{(S_4)}
T_{(3,1)}
\left(U_{\mathrm{QCT}}^{(S_4)}\right)^\dagger
=
\operatorname{diag}
\left(
1,1,0,0,-\frac12
\right),
$
so the sectors with eigenvalue \(1\) are preserved, those with eigenvalue \(0\) are annihilated, and the remaining sector is multiplied by \(-1/2\). 

\subsubsection{Spectral Implementation through the QCT}

Access to the quantum character transform provides a second implementation strategy in which a class convolution operator is specified by its eigenvalues instead of its expansion in group actions. In the irreducible-character basis, \(T_f\) is diagonal, with entries given by the spectral coefficients \(\widehat f(\Gamma)\), which is useful for operations defined spectrally. The spectral representation can also improve the block-encoding normalization, since it depends on \(\max_\Gamma|\widehat f(\Gamma)|\) instead of the \(\ell_1\) weight of the class-sum decomposition.

Let
\[
\beta_f
\geq
\max_{\Gamma\in\widehat{\mathcal G}}
\left|\widehat f(\Gamma)\right|
=
\|T_f\|.
\]
Let \(\widehat f_\epsilon(\Gamma)\) be a finite-precision
encoding of \(\widehat f(\Gamma)\), chosen such that
$
\max_{\Gamma\in\widehat{\mathcal G}}
\left|
\widehat f_\epsilon(\Gamma)-\widehat f(\Gamma)
\right|
\leq
\epsilon\beta_f,
\qquad
\left|\widehat f_\epsilon(\Gamma)\right|\leq\beta_f.
$
Define the reversible spectral oracle
$
O_{\widehat f}:
|\Gamma\rangle|y\rangle
\longmapsto
|\Gamma\rangle
|y\oplus\widehat f_\epsilon(\Gamma)\rangle,
$
so that
$
O_{\widehat f}|\Gamma\rangle|0\rangle
=
|\Gamma\rangle|\widehat f_\epsilon(\Gamma)\rangle.
$
Using reversible arithmetic, condition an ancilla
amplitude-and-phase operation on the encoded value so that
$
(\langle 0|_{\mathrm{anc}}\otimes I)
U_{D_{\widehat f_\epsilon}}
(|0\rangle_{\mathrm{anc}}\otimes I)
=
\frac{D_{\widehat f_\epsilon}}{\beta_f},
$
where the ancilla contains the spectral-value and signal registers.
Conjugating this block encoding by the QCT,
$
U_f^{\mathrm{spec}}
=
\left(
I\otimes
\left(U_{\mathrm{QCT}}^{(\mathcal G)}\right)^\dagger
\right)
U_{D_{\widehat f_\epsilon}}
\left(
I\otimes U_{\mathrm{QCT}}^{(\mathcal G)}
\right),
$
which satisfies
$
\left\|
(\langle 0|_{\mathrm{anc}}\otimes I)
U_f^{\mathrm{spec}}
(|0\rangle_{\mathrm{anc}}\otimes I)
-
\frac{T_f}{\beta_f}
\right\|
\leq\epsilon.
$
\(U_f^{\mathrm{spec}}\) is an \(\epsilon\)-approximate block encoding of \(T_f\) with spectral normalization \(\beta_f\).

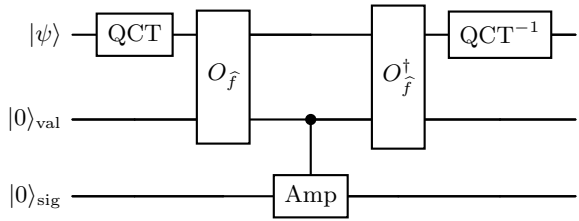
\begin{figure}[H]
\centering
\begin{quantikz}[row sep=0.35cm,column sep=0.32cm]
\lstick{$|\psi\rangle$}
& \gate{\mathrm{QCT}}
& \gate[2]{O_{\widehat f}}
& \qw
& \gate[2]{O_{\widehat f}^{\dagger}}
& \gate{\mathrm{QCT}^{-1}}
& \qw \\
\lstick{$|0\rangle_{\mathrm{val}}$}
& \qw
& \qw
& \ctrl{1}
& \qw
& \qw
& \qw \\
\lstick{$|0\rangle_{\mathrm{sig}}$}
& \qw
& \qw
& \gate{\mathrm{Amp}}
& \qw
& \qw
& \qw
\end{quantikz}
\caption{Spectral implementation of a class convolution operator. The QCT maps the input into the irreducible-character basis, where the oracle $O_{\widehat f}$ computes the spectral coefficients $\widehat f(\Gamma)$. The gate $\mathrm{Amp}$ performs a controlled amplitude encoding using the normalized spectral values $\widehat f(\Gamma)/\beta_f$. After uncomputing the spectral register and applying \(\mathrm{QCT}^{-1}\), the circuit is an \(\epsilon\)-approximate block encoding of \(T_f/\beta_f\) on the input register, relative to the \(\lvert 0\rangle_{\mathrm{val}}\lvert 0\rangle_{\mathrm{sig}}\) ancilla state.
}
\label{fig:spectral_qct}
\end{figure}

\begin{proposition}[Spectral block encoding through the QCT]
Let \(f\in \mathrm{Cl}(\mathcal{G})\), and
$
\beta_f
\geq
\max_{\Gamma\in\widehat{\mathcal G}}
|\widehat f(\Gamma)|.
$
Then \(T_f/\beta_f\) has an \(\epsilon\)-approximate block encoding with complexity
\[
2\operatorname{cost}(\mathrm{QCT})
+
2\operatorname{cost}(O_{\widehat f})
+
\operatorname{polylog}(1/\epsilon),
\]
where \(\epsilon\) is the precision of the diagonal block encoding.
\end{proposition}

Choosing the optimal spectral normalization,
$
\beta_f
=
\max_{\Gamma\in\widehat{\mathcal G}}
|\widehat f(\Gamma)|
=
\|T_f\|,
$
\[
\beta_f
\leq
\sum_i |C_i|\,|f(C_i)|
=
\alpha_f.
\]
The inequality may be strict, so the spectral implementation can have a smaller normalization than the direct LCU construction.

The same spectral construction applies to functions of \(T_f\). For a computable scalar function \(\phi\),
\[
\phi(T_f)
=
U_{\mathrm{QCT}}^\dagger
\operatorname{diag}\!\left(
\phi(\widehat f(\Gamma))
\right)
U_{\mathrm{QCT}}.
\]
If 
$
\phi(x)=e^{-itx},
$
\[
e^{-itT_f}
=
U_{\mathrm{QCT}}^\dagger
\operatorname{diag}\!\left(
e^{-it\widehat f(\Gamma)}
\right)
U_{\mathrm{QCT}}.
\]
Hamiltonian simulation can be performed by replacing the amplitude-encoding operation in Fig.~\ref{fig:spectral_qct} with the diagonal phase operation
$
U_{\mathrm{phase}}(t)|\Gamma\rangle
=
e^{-it\widehat f(\Gamma)}|\Gamma\rangle.
$

\begin{proposition}[Hamiltonian simulation through the QCT]
Let \(T_f\) be Hermitian. Given coherent access to the spectral coefficients through \(O_{\widehat f}\), the evolution \(e^{-itT_f}\) can be implemented to precision \(\epsilon\) with complexity
\[
2\operatorname{cost}(\mathrm{QCT})
+
2\operatorname{cost}(O_{\widehat f})
+
\operatorname{cost}(U_{\mathrm{rot}}(t,\epsilon)).
\]
\end{proposition}
The final term accounts for the reversible arithmetic and phase-rotation circuit applying \(U_{\mathrm{phase}}(t)\) 
to precision \(\epsilon\) once \(\widehat f(\Gamma)\) has been computed. Quantum phase estimation is not required; the QCT supplies the eigenbasis, while \(O_{\widehat f}\) supplies the eigenvalues. Other efficiently computable unit-modulus spectral functions can be implemented by replacing the phase \(e^{-itx}\), while bounded spectral functions that are not unit modulus require an appropriate diagonal block encoding or a polynomial-transformation procedure such as quantum singular value transformation~\cite{Gilyen2019QSVT}.

\section{Algorithm by Restriction}

The Restriction Theorem facilitates a universal implementation strategy for the QCT. By Theorem~\ref{thm:center-restriction}, the non-abelian Fourier transform maps the center into the invariant lines
\[
\mathscr F_{\mathcal G}\ket{\psi}
=
\sum_{\Gamma\in\widehat{\mathcal G}}
\widehat{\psi}(\Gamma)\,
\ket{\Gamma}\ket{I_\Gamma},
\]
where
$
\ket{I_\Gamma}
=
\frac{1}{\sqrt{d_\Gamma}}
\operatorname{vec}(I_{V_\Gamma})
=
\frac{1}{\sqrt{d_\Gamma}}
\sum_{i=1}^{d_\Gamma}\ket{i,i}
$
is the canonical maximally entangled state in Theorem~\ref{thm:center-restriction}. Thus, on the Fourier image of the center, the matrix-index registers contain no independent information about $\ket{\psi}$; conditioned on $\Gamma$, they occupy the known state $\ket{I_\Gamma}$. A restriction algorithm coherently unprepares these states, leaving only the irreducible-character amplitudes.

\begin{definition}[Restriction operator]
Define the \textit{restriction operator}
\[
\mathscr R_{\mathcal G}
:=
\sum_{\Gamma\in\widehat{\mathcal G}}
\ket{\Gamma}\!\bra{\Gamma}
\otimes
\bra{I_\Gamma}.
\]
In the irreducible-block decomposition of the Fourier space,
$
\mathscr R_{\mathcal G}
=
\operatorname{diag}
\bigl(
\bra{I_{\Gamma_1}},
\ldots,
\bra{I_{\Gamma_d}}
\bigr),
$
and its action on the invariant subspace is
\[
\mathscr R_{\mathcal G}
\bigl(\ket{\Gamma}\ket{I_\Gamma}\bigr)
=
\ket{\Gamma}.
\]
\end{definition}

The restriction operator is a partial isometry from the full Fourier
space onto the character space. In particular,
$
\mathscr R_{\mathcal G}
\mathscr R_{\mathcal G}^{\dagger}
=
I_{|\widehat{\mathcal G}|},
$
while
$
\mathscr R_{\mathcal G}^{\dagger}
\mathscr R_{\mathcal G}
=
\sum_{\Gamma\in\widehat{\mathcal G}}
\ket{\Gamma}\!\bra{\Gamma}
\otimes
\ket{I_\Gamma}\!\bra{I_\Gamma},
$
which is the orthogonal projector onto
$
\operatorname{Im}\!\left(
\left.
\mathscr F_{\mathcal G}
\right|_{Z(\mathbb C[\mathcal G])}
\right)
=
\bigoplus_{\Gamma\in\widehat{\mathcal G}}
\mathbb C\,\ket{\Gamma}\ket{I_\Gamma}.
$
On this subspace, $\mathscr R_{\mathcal G}$ is the canonical unitary
identification with the character space established by the Restriction
Theorem. Consequently,
\[
U_{\rm QCT}
=
\mathscr R_{\mathcal G}
\mathscr F_{\mathcal G}
\big|_{Z(\mathbb C[\mathcal G])}.
\]
To implement this identification coherently, for each positive integer $d$ we need a maximally-entangled-state preparation unitary $W_d$ satisfying
\begin{equation}
W_d\ket{0,0}
=
\ket{I_d}
:=
\frac{1}{\sqrt d}
\sum_{i=0}^{d-1}\ket{i,i}.   
\end{equation}
In a logical $d$-dimensional encoding, such a unitary can be compiled by applying the Fourier transform $F_d$ to the first matrix-index register, followed by the controlled-addition operation
$
\operatorname{SUM}_d\ket{i,j}
=
\ket{i,j+i \bmod d},
$
so that
\begin{equation}
W_d=\operatorname{SUM}_d(F_d\otimes I).
\end{equation}
Then,
$
W_d^\dagger\ket{I_d}
=
(F_d^\dagger\otimes I)\operatorname{SUM}_d^\dagger\ket{I_d}
=
\ket{0,0}.
$
For $d=2$, $F_2=H$ and
$\operatorname{SUM}_2=\operatorname{CNOT}$, and one can use
$
W_2=\operatorname{CNOT}(H\otimes I),
$
the standard Bell-state preparation circuit. Taking the matrix-index registers to have a common ambient dimension, extend each \(W_{d_\Gamma}\) arbitrarily to a unitary on these registers. The resulting controlled unitary is
\begin{equation}
\mathscr R_{\mathcal G}^{W}
:=
\sum_{\Gamma\in\widehat{\mathcal G}}
\ket{\Gamma}\!\bra{\Gamma}
\otimes
W_{d_\Gamma}^{\dagger}.
\end{equation}
It satisfies
$
\mathscr R_{\mathcal G}^{W}
\bigl(\ket{\Gamma}\ket{I_\Gamma}\bigr)
=
\ket{\Gamma}\ket{0,0},
$
and hence, for every
$\ket{\psi}\in Z(\mathbb C[\mathcal G])$,
$
\mathscr R_{\mathcal G}^{W}
\mathscr F_{\mathcal G}\ket{\psi}
=
\bigl(U_{\rm QCT}\ket{\psi}\bigr)\ket{0,0}.
$
The matrix-index registers are returned to the fixed state $\ket{0,0}$ and can be discarded without loss of coherence.

\begin{algorithm}
\centering
\textbf{Restriction Algorithm for $\mathrm{QCT}_{\mathcal G}$}
\par\medskip
\begin{algorithmic}[1]
\State {\bf Input:} State
\(\ket{\psi}\in Z(\mathbb C[\mathcal G])\), encoded in the full group register as a superposition of normalized conjugacy-class states.
\State Apply the group QFT:
$
\ket{\phi}=\mathscr F_{\mathcal G}\ket{\psi}.
$
\State Apply the controlled restriction:
\[
\mathscr R_{\mathcal G}^{W}\ket{\phi}
=
(U_{\rm QCT}\ket{\psi})\ket{0,0}.
\]
\State Discard the matrix-index registers
$\ket{0,0}$.
\State {\bf Output:}
$U_{\rm QCT}\ket{\psi}$ in the irreducible-character basis.
\end{algorithmic}
\end{algorithm}

The additional cost of the restriction step is that of the controlled
inverses $W_{d_\Gamma}^{\dagger}$, so
\[
\operatorname{cost}({\rm QCT}_{\mathcal G})
=
\operatorname{cost}({\rm QFT}_{\mathcal G})
+
\operatorname{cost}(\mathscr R_{\mathcal G}^{W}).
\]
The restriction overhead is the cost of implementing
\(W_{d_\Gamma}^{\dagger}\) coherently from the irrep label \(\Gamma\). Because \(d_\Gamma^2\leq|\mathcal G|\), the joint matrix-index space has dimension at most \(|\mathcal G|\), so the restriction step uses
\(\mathcal O(\log|\mathcal G|)\) matrix-index qubits under a binary encoding. For families of finite groups in which \(d_\Gamma\) can be efficiently computed from \(\Gamma\), the restriction step reduces to
coherently controlled generalized Bell-state unpreparation and can be implemented with complexity polynomial in \(\log|\mathcal G|\) and
\(\log(1/\epsilon)\). This holds for finite-group families with an efficient QFT and access to representation data, such as \(D_n\), whose
irreducible representations have dimension one or two, and \(S_n\), where irreducible representations \(\Gamma\) are indexed by partitions \(\lambda\vdash n\) and \(d_\Gamma\) is efficiently computable from the
hook-length formula.

While the restriction algorithm operates on the full group space, its output has dimension $k(\mathcal G)=|\widehat{\mathcal G}|$, with one label per irreducible representation. A full group register uses $\lceil\log_2|\mathcal G|\rceil$ qubits, while the character register requires only $\lceil\log_2 k(\mathcal G)\rceil$. This reduction provides a natural space-complexity parameter for comparing a \textit{class-native QCT} with the restriction of a full group Fourier transform. When $k(\mathcal G)$ grows substantially more slowly than $|\mathcal G|$, a QCT algorithm that works completely in the class basis acts on an asymptotically smaller state space, meaning smaller group-specific transform circuits may be possible.

\begin{remark}[Central compression]
For a finite group \(\mathcal G\), the \textit{commuting probability} is
\[
\operatorname{cp}(\mathcal G)
=
\frac{|\{(g,h)\in\mathcal G^2:gh=hg\}|}{|\mathcal G|^2}
=
\frac{k(\mathcal G)}{|\mathcal G|},
\]
where \(k(\mathcal G)=|\operatorname{Conj}(\mathcal G)|\). Since
$
\dim Z(\mathbb C[\mathcal G])=k(\mathcal G),
\dim\mathbb C[\mathcal G]=|\mathcal G|,
$
it follows that
\[
\operatorname{cp}(\mathcal G)
=
\frac{\dim Z(\mathbb C[\mathcal G])}
{\dim\mathbb C[\mathcal G]}.
\]
The ratio between the QCT and QFT state-space dimensions is the commuting probability of the group, with
$
\frac{|\mathcal G|}{k(\mathcal G)}
=
\operatorname{cp}(\mathcal G)^{-1}
$
being the state-space compression factor for non-abelian groups.
\end{remark}

For metacyclic groups such as \(D_n\), however, this dimensional reduction does not produce an asymptotic reduction in register size. The dihedral group \(D_n\) has order \(2n\), while
\[
k(D_n)
=
\begin{cases}
(n+3)/2, & n \text{ odd},\\
n/2+3, & n \text{ even},
\end{cases}
\]
so \(k(D_n)=n/2+ \mathcal O(1)\).
$
\log_2|D_n|-\log_2 k(D_n)\to 2,
$
and the class representation saves asymptotically only two qubits. Efficient quantum Fourier transforms for metacyclic groups are already known~\cite{hoyer1997efficient} with polylogarithmic gate complexity for \(D_n\), so a constant reduction in register size shows no immediate indication of an asymptotically faster class-native character transform, although the existence of such a quantum circuit has not been formally ruled out. The classical complexity follows similarly. Individual irreducible characters of \(D_n\) have explicit closed forms and can be evaluated in \(\mathcal O(1)\) arithmetic operations, while materializing the normalized character matrix requires writing \(\Theta(n^2)\) entries. A classical restriction algorithm instead embeds the class representation into the group representation, applies a fast Fourier transform, and compresses the result back to the character representation. The embedding sends a class coefficient \(a_i\) to the constant value \(a_i/\sqrt{|C_i|}\) on every element of \(C_i\), where expanding the class representation over \(D_n\) requires \(\Theta(|D_n|)=\Theta(n)\) operations, and compression is likewise \(\Theta(n)\). Fast Fourier transforms for dihedral groups run in time \(\mathcal O(|D_n|\log|D_n|)= \mathcal O(n\log n)\)~\cite{clausen1989fast}, so the embedding and compression costs do not change the \(\mathcal O(n\log n)\) complexity of the classical restriction algorithm.

The situation is different for the symmetric group \(S_n\). Its order is \(n!\), but \(k(S_n)=p(n)\), where
\(p(n)\sim \frac{1}{4n\sqrt{3}}\exp(\pi\sqrt{2n/3})\) is the partition number. The group register is of size
\(\log|S_n|=\Theta(n\log n)\), whereas \(\log k(S_n)=\Theta(\sqrt n)\), a nontrivial asymptotic gap. A class-native algorithm can provide a polynomial compression in register size. In 2016, Kawano and Sekigawa reported an \(\mathcal O(n^3\log n)\) algorithm for the \(S_n\) quantum Fourier transform~\cite{KawanoSekigawa2016}, and a recent reanalysis by Bruinsma, Grinko, and Ozols improves the gate bound to \(\widetilde{\mathcal O}(n^3)\)~\cite{BruinsmaGrinkoOzols2026}. Shortly after the initial dissemination of this manuscript, Rathore and Zhou shared the first known algorithm for the QCT over \(S_n\) without passing through the full group QFT, using \(\widetilde{\mathcal O}(n^{5/2})\) gates and \(\mathcal O(n)\) qubits ~\cite{RathoreZhou2026}. While the restriction algorithm inherits the complexity of the QFT and the new Hamiltonian-based method achieves a square-root-factor improvement in gate complexity while using \(\mathcal O(n)\) qubits, a native QCT acting entirely on this compressed class space could provide a further polynomial reduction in register size. Whether that space reduction can also yield a lower asymptotic gate complexity remains the more important question, with potential consequences for class convolution, symmetry-adapted spectral algorithms, and conjugacy-invariant formulations of graph isomorphism.

\section{Hidden Subgroups and Central Shadows}

One of the primary motivations discussed for this transform was its application to the hidden subgroup problem over non-abelian groups, which turns out to be quite limited. This result is not surprising, as the restriction algorithm provides a rigorous way to understand the information retained by the QCT in hidden-subgroup problems. In the most general hidden subgroup problem, an oracle is constant on left cosets of an unknown subgroup \(H\leq\mathcal G\) and factors through the left-coset space \(\mathcal G/H\) ~\cite{Hoyer2002,EttingerHoyerKnill2004}, while the data relevant to the QCT is conjugacy invariant and factors through the set of conjugacy classes \(\mathcal G/\!/\mathcal G\). It is obvious that the character transform cannot reproduce the full information available under nonabelian Fourier sampling, but which information about \(H\) survives restriction to the center \(Z(\mathbb C[\mathcal G])\)? We call the invariant the \emph{central shadow} of \(H\), and it presents both a positive and a negative finding. It identifies the subgroup information accessible to class-convolution operators, while showing that distinct nonconjugate subgroups with the same central shadow cannot be distinguished by any protocol whose dependence on the hidden subgroup factors entirely through the class algebra.

For each irreducible representation \(\rho_\Gamma:\mathcal G\rightarrow U(V_\Gamma)\), the normalized subgroup average determines the operator
\begin{equation}
\Pi_H^{(\Gamma)}
=
\frac{1}{|H|}
\sum_{h\in H}\rho_\Gamma(h).
\end{equation}
Since averaging over \(H\) is idempotent and self-adjoint in a unitary representation, \(\Pi_H^{(\Gamma)}\) is the orthogonal projector onto the \(H\)-invariant subspace ~\cite{Serre1977}
\[
V_\Gamma^H
=
\{v\in V_\Gamma:\rho_\Gamma(h)v=v\ \text{for all }h\in H\}.
\]
The full nonabelian Fourier description of the subgroup average retains, within each irreducible sector \(\Gamma\), the matrix-valued projector \(\Pi_H^{(\Gamma)}\). Its rank is
\[
\operatorname{rank}\Pi_H^{(\Gamma)}
=
\operatorname{Tr}\Pi_H^{(\Gamma)}
=
\dim V_\Gamma^H.
\]
Projecting the subgroup average onto \(Z(\mathbb C[\mathcal G])\) removes this internal matrix structure while retaining its conjugacy-invariant content. 

\begin{definition}[Central shadow of a subgroup]
Let \(H\leq\mathcal G\). The \emph{central shadow} of \(H\) is the
conjugacy average of its normalized subgroup sum,
\begin{align*}
z_H
&:=
\frac{1}{|\mathcal G|}
\sum_{x\in\mathcal G}
x
\left(
\frac{1}{|H|}\sum_{h\in H}h
\right)
x^{-1}\\
&=
\sum_{C\in\operatorname{Conj}(\mathcal G)}
\frac{|H\cap C|}{|H|\,|C|}\,K_C,
\qquad
K_C:=\sum_{g\in C}g.
\end{align*}
\end{definition}

In the class basis, \(z_H\) records the distribution of \(H\)
across the conjugacy classes of \(\mathcal G\). It forgets which
individual elements occur within a class and how those elements
are jointly embedded as a subgroup.

\begin{theorem}[Character spectrum of a subgroup central shadow]
\label{thm:central-shadow}
For every \(H\leq\mathcal G\) and
\(\Gamma\in\widehat{\mathcal G}\), define
$
m_\Gamma(H)
:=
\frac{1}{|H|}
\sum_{h\in H}\chi_\Gamma(h).
$
Then
\[
\begin{aligned}
m_\Gamma(H)
&=\dim V_\Gamma^H\\
&=\left\langle
\operatorname{Res}^{\mathcal G}_H\chi_\Gamma,1_H
\right\rangle_H
=\left\langle
\chi_\Gamma,
\operatorname{Ind}^{\mathcal G}_H1_H
\right\rangle_{\mathcal G}.
\end{aligned}
\]
where \(1_H\) denotes the trivial character of \(H\). Moreover,
\[
\rho_\Gamma(z_H)
=
\frac{m_\Gamma(H)}{d_\Gamma}I_{d_\Gamma},
\]
and hence, under the spectral convention for class convolution,
\[
\widehat z_H(\Gamma)
=
\frac{m_\Gamma(H)}{d_\Gamma}
=
\frac{\dim V_\Gamma^H}{d_\Gamma}.
\]
The multiplicity profile $\bigl(m_\Gamma(H)\bigr)_{\Gamma\in\widehat{\mathcal G}}
$
is the irreducible-multiplicity vector of
$
\mathbb C[\mathcal G/H]
\cong
\operatorname{Ind}^{\mathcal G}_H1_H.
$
\end{theorem}
This also makes clear what is lost under the restriction \(\mathscr R_{\mathcal G}\). Before restriction, the Fourier block associated with \(\Gamma\) contains the projector \(\Pi_H^{(\Gamma)}\), whereas its central component is captured only by the scalar \(\operatorname{Tr}\Pi_H^{(\Gamma)}\). At the level of the \(\Gamma\)-block, centralization acts as
\[
\Pi_H^{(\Gamma)}
\longmapsto
\frac{\operatorname{Tr}\Pi_H^{(\Gamma)}}{d_\Gamma}I_{d_\Gamma}
=
\frac{m_\Gamma(H)}{d_\Gamma}I_{d_\Gamma}.
\]
The character transform retains the dimension, but not the orientation of the invariant subspace \(V_\Gamma^H\) within \(V_\Gamma\). That missing information is encoded in the internal matrix indices present in \(\mathscr F_{\mathcal G}\) and removed by restriction to the center.

\subsection{Weak Fourier Sampling}

The same invariant appears in weak Fourier sampling. For a left coset state
$
|gH\rangle
=
\frac{1}{\sqrt{|H|}}
\sum_{h\in H}|gh\rangle,
$
application of the nonabelian Fourier transform
\(\mathscr F_{\mathcal G}\), followed by measurement of only the
irreducible-representation label \(\Gamma\), gives
\[
\Pr_H(\Gamma)
=
\frac{|H|d_\Gamma}{|\mathcal G|}
\operatorname{rank}\Pi_H^{(\Gamma)}
=
\frac{|H|d_\Gamma}{|\mathcal G|}
\dim V_\Gamma^H.
\]
Thus weak Fourier sampling depends on \(H\) only through the multiplicity profile
\[
m_H
=
\bigl(m_\Gamma(H)\bigr)_{\Gamma\in\widehat{\mathcal G}}
=
\bigl(\dim V_\Gamma^H\bigr)_{\Gamma\in\widehat{\mathcal G}}.
\]
Equivalently,
$
\Pr_H(\Gamma)
=
\frac{|H|d_\Gamma^2}{|\mathcal G|}
\widehat z_H(\Gamma).
$
The distinction from strong Fourier sampling is that weak sampling retains only the irreducible-sector weights, whereas measurements of the internal Fourier indices may additionally probe the matrix structure of \(\Pi_H^{(\Gamma)}\) removed by \(\mathscr R_{\mathcal G}\).

The class-intersection profile, central shadow, and multiplicity profile are three descriptions of the same conjugacy-invariant information. For two subgroups \(H_1,H_2\leq\mathcal G\),
\[
z_{H_1}=z_{H_2}
\Longleftrightarrow
\frac{|H_1\cap C|}{|H_1|}
=
\frac{|H_2\cap C|}{|H_2|}
\quad
\forall C\in\operatorname{Conj}(\mathcal G),
\]
and, since the QCT is invertible on the center, equivalently
\[
\frac{1}{|H_1|}
\sum_{h\in H_1}\chi_\Gamma(h)
=
\frac{1}{|H_2|}
\sum_{g\in H_2}\chi_\Gamma(g)
\qquad
\forall\Gamma\in\widehat{\mathcal G}.
\]
Taking \(C=\{e\}\) in the class-profile equality gives
$
|H_1|=|H_2|,
$
so equality of central shadows is equivalent to
\[
|H_1\cap C|
=
|H_2\cap C|
\qquad
\forall C\in\operatorname{Conj}(\mathcal G).
\]
This is the classical relation of \emph{almost conjugacy}, or Gassmann equivalence ~\cite{Gassmann1926,Perlis1977,Sunada1985}. Equivalently,
$
\operatorname{Ind}^{\mathcal G}_{H_1}1_{H_1}
\cong
\operatorname{Ind}^{\mathcal G}_{H_2}1_{H_2},
$
and hence
$
\dim V_\Gamma^{H_1}
=
\dim V_\Gamma^{H_2}
$
for every \(\Gamma\in\widehat{\mathcal G}\) by Frobenius reciprocity~\cite{Serre1977,Perlis1977,Sunada1985}. Thus
\[
z_{H_1}=z_{H_2}
\;\Longleftrightarrow\;
H_1\ \text{and}\ H_2\ \text{are Gassmann equivalent}.
\]

\begin{corollary}[Central-shadow indistinguishability]
\label{cor:central-shadow-indistinguishability}
Let \(H_1,H_2\leq\mathcal G\) satisfy
$
z_{H_1}=z_{H_2}.
$
Then their character spectra coincide, and every class-convolution operator has identical spectral response on the two subgroup shadows. More generally, any procedure whose dependence on the hidden subgroup factors entirely through \(Z(\mathbb C[\mathcal G])\) cannot distinguish \(H_1\) from \(H_2\).
\end{corollary}

This limitation cannot be overcome by changing the class function, combining multiple class-convolution operators, or selecting such operators adaptively. Optimizing a central probe can change how the surviving information is accessed, but cannot recover information removed by projection onto \(Z(\mathbb C[\mathcal G])\). Nonconjugate Gassmann-equivalent subgroups give intrinsic central-shadow collisions, not artifacts of a particular circuit or measurement prescription.

\subsection{Hidden Subgroup Conjugacy Problem}
A problem especially well matched to the conjugacy invariance of the central shadow is the \emph{hidden subgroup conjugacy problem} (HSCP). In the HSCP, an oracle hides a subgroup \(H\leq\mathcal G\) as in the ordinary hidden subgroup problem, but the objective is only to know its conjugacy class
\[
[H]_{\mathcal G}
=
\{gHg^{-1}:g\in\mathcal G\},
\]
instead of the particular hidden subgroup~\cite{Bacon2008}. This is exactly the symmetry respected by the central shadow, since
$
z_{xHx^{-1}}=z_H
$
for every \(x\in\mathcal G\).

The central shadow therefore defines a conjugacy invariant of the hidden subgroup. The relevant question for the HSCP is whether this invariant is complete on the family of subgroups under consideration.

\begin{proposition}[Central-shadow criterion for the HSCP]
Let \(\mathcal H\) be a promised family of subgroups of \(\mathcal G\). The QCT central shadow is a complete invariant for the HSCP on \(\mathcal H\) if and only if
\[
H_1, H_2\in\mathcal H,
\qquad
z_{H_1}=z_{H_2}
\quad\Longrightarrow\quad
H_1\sim_{\mathcal G}H_2.
\]
Equivalently, \(\mathcal H\) contains no pair of nonconjugate Gassmann-equivalent subgroups.
\end{proposition}

\begin{proof}
Conjugate subgroups have identical central shadows. Since the QCT is invertible on \(Z(\mathbb C[\mathcal G])\), it preserves equality and inequality of central shadows. Thus the QCT central shadow reveals
\([H]_{\mathcal G}\) on \(\mathcal H\) exactly when
$
[H]_{\mathcal G}\mapsto z_H
$
is injective on the promised family, or equivalently when
\(\mathcal H\) contains no pair of nonconjugate Gassmann-equivalent
subgroups.
\end{proof}

The criterion admits both positive and negative examples even among small symmetric groups. A simple positive example occurs for the family of Klein four subgroups of \(S_4\). Consider the Klein four subgroups
\begin{gather*}
H_1= \{e,(12)(34),(13)(24),(14)(23)\}
\\
H_2=\{e,(12),(34),(12)(34)\}.  
\end{gather*}
Their class-intersection profiles are $(1,0,3,0,0)$ and $(1,2,1,0,0),$ respectively. They are nonconjugate, and
$
z_{H_1}\neq z_{H_2},
$
so their subgroup conjugacy classes are distinguishable by the central shadow.

However, this separation fails in general.  In $S_6$,  consider
\begin{align*}
H_1
&=
\{
e,\,
(34)(56),\,
(35)(46),\,
(36)(45)
\},\\
H_2
&=
\{
e,\,
(34)(56),\,
(12)(56),\,
(12)(34)
\}.
\end{align*}
Each subgroup intersects the identity class once and the class of cycle type
\(2^2 1^2\) three times, with empty intersection with every other conjugacy
class. So their central shadows coincide:
$z_{H_1}=z_{H_2}.$ The subgroups are Gassmann equivalent but nevertheless nonconjugate. Their orbit-size decompositions in the natural action on six points are respectively $1+1+4$ and $2+2+2$, which are invariant under relabeling. 

This also places the QCT relative to existing HSCP and character-sampling algorithms. Writing
$
\rho_H
=
\frac{1}{|\mathcal G|}
\sum_{x\in\mathcal G}
\ket{xH}\!\bra{xH}
$
for the single-copy hidden-subgroup state, Bacon showed that the single-copy hidden-subgroup conjugacy state
$
\rho_{[H]}
=
\frac{1}{|\mathcal G|}
\sum_{g\in\mathcal G}\rho_{gHg^{-1}}
$
decomposes after Fourier transformation into scalar irreducible blocks whose coefficients contain the same character averages $\sum_{h\in H}\chi_\Gamma(h)$~\cite{Bacon2008}.  At the single-copy level, the central shadow captures the conjugacy-invariant multiplicity data.  The multicopy problem is richer: the conjugacy average of several copies retains correlations between copies, and Bacon's Heisenberg-group algorithm uses a Clebsch--Gordan transform and its multiplicity register to solve the HSCP with two hidden-subgroup states.

For $S_n$, Bravyi, Gosset, Havlicek, and Schatzki construct a polynomial-size quantum circuit that prepares the uniform state over the conjugacy class $C_g$ of a specified permutation $g$ and, after the $S_n$ QFT, obtains character amplitudes proportional to $\chi_\Gamma(g)$ ~\cite{Bravyi2025Characters}. Their circuit exploits the same scalarization of central class sums underlying the QCT, but addresses a different task: it takes a specified conjugacy class as input, whereas the HSCP begins with an unknown hidden subgroup. The circuit alone does not solve the HSCP.

The proposition above is information-theoretic; by itself, it makes no claim about efficiency.  Even when distinct promised subgroup conjugacy classes have distinct central shadows, an efficient HSCP algorithm further requires efficient access to the shadow, an efficient QCT, and sufficient separation of the resulting spectral data to permit discrimination with polynomial resources.  Conversely, a Gassmann collision is absolute for central methods. No improvement in spectral probing or measurement can recover information already removed by projection onto $Z(\mathbb C[\mathcal G])$.  Multicopy transforms such as the Clebsch--Gordan construction lie beyond this one-copy central restriction and may access additional representation-theoretic correlations.

\subsection{Graph Isomorphism}

Graph Isomorphism provides an important test case for the limitations of
central information. The standard reduction of Graph Isomorphism to a
hidden-subgroup problem produces a nonabelian HSP over a
symmetric-group-related permutation group. This problem was one of the
principal motivations for the study of nonabelian Fourier sampling, but it
is also a setting in which some of its strongest limitations are
known~\cite{MooreRussellSchulman2008,HallgrenEtAl2010}.

For two rigid graphs on \(n\) vertices, the reduction is formulated over the wreath product
\[
\mathcal G
=
S_n\wr S_2
=
(S_n\times S_n)\rtimes S_2.
\]
In the nonisomorphic case the hidden subgroup is trivial, whereas in the isomorphic case it is generated by an involution exchanging the two graph
copies according to an isomorphism between their vertex sets. Under the standard embedding into \(S_{2n}\), these generators are fixed-point-free
involutions and hence have a definite conjugacy-class signature. This makes the reduction superficially well matched to a character-based analysis. The difficulty, however, is that this conjugacy-class signature does not make the two HSP hypotheses efficiently distinguishable by single-register Fourier information, much less identify the underlying graph isomorphism.

Indeed, Moore, Russell, and Schulman showed that the Fourier-sampling framework does not provide sufficient information for the Graph-Isomorphism-related hidden-subgroup problem over symmetric groups~\cite{MooreRussellSchulman2008}. Their result rules out even strong Fourier sampling in this setting; the obstruction is information-theoretic and cannot be overcome by choosing a different basis within an irreducible block. Since the QCT retains only the conjugacy-invariant component of this representation-theoretic information, the central restriction cannot bypass the Fourier-sampling barrier, and the limitation extends to joint measurements on a bounded number of coset states. Hallgren, Moore, R"otteler, Russell, and Sen showed that obtaining non-negligible information for the Graph Isomorphism HSP requires measurements entangled across $\Omega(n\log n)$ coset states~\cite{HallgrenEtAl2010}, so the missing information can reside in correlations among many hidden-subgroup states and remain inaccessible to any single-copy central statistic. Nor does the need for multicopy information automatically make representation-theoretic sieving sufficient. Moore, Russell, and \'{S}niady showed that a broad class of Kuperberg-style quantum sieve algorithms, which iteratively combine representations according to tensor-product decompositions, cannot yield a polynomial-time algorithm for the Graph Isomorphism HSP; within their model the required running time is \(e^{\Omega(\sqrt n)}\)~\cite{MooreRussellSniady2010}.

These results clarify the scope of the QCT rather than diminishing its role. Full subgroup recovery and central property testing are different computational tasks, and Graph Isomorphism demonstrates that the multiplicity data retained by the central shadow does not suffice for the full hidden-subgroup problem. For families with distinct central shadows, class-convolution operators may nevertheless be used as discriminating central probes by weighting the irreducible sectors differently. Such spectral processing can increase the separation between distinct shadows but, by Corollary~\ref{cor:central-shadow-indistinguishability}, cannot distinguish subgroups whose central shadows coincide.

The Graph Isomorphism obstruction also indicates where additional algebraic structure must enter. Passing beyond \(Z(\mathbb C[\mathcal G])\) to noncentral operator algebras can retain internal representation and relational information erased by centralization. For permutation actions, orbital and coherent algebras provide nice examples of such extensions, connecting representation-theoretic operators with coherent configurations and Weisfeiler--Leman-type refinement. Developing transforms for these larger algebras is beyond the scope here, but the central case establishes the conjugacy-invariant endpoint of what survives centralization and when noncommutativity may become necessary.
\section{Applications: Graphs, Walks, and Symmetric Hamiltonians}

The connection between Fourier transforms and graph spectra is already present in the abelian setting. A circulant graph on \(n\) vertices is a Cayley graph of the cyclic group \(\mathbb Z_n\), and its adjacency matrix is circulant and diagonalized by the discrete Fourier transform~\cite{Terras1999}, or on a quantum computer, the quantum Fourier transform. Generally, for a finite abelian group \(\mathcal G\), translation-invariant graph operators are convolution operators on \(\mathbb C[\mathcal G]\) and are diagonalized by the group Fourier transform; the eigenvalues of an abelian Cayley graph are the character sums over its connection set~\cite{Terras1999,Babai1979}. For a non-abelian group, ordinary convolution is no longer diagonalized by scalar characters, and the non-abelian Fourier transform decomposes into non-commutative matrix-valued blocks. However, a distinguished commutative sector survives. If the connection weights are invariant under conjugation, the convolution operator belongs to the center $Z(\mathbb{C}[\mathcal{G}])$. The passage from QFT on circulant graphs to QCT on conjugacy-invariant Cayley graphs extends the familiar Fourier description of cyclic graph spectra from abelian translation symmetry to non-abelian conjugacy symmetry. 

\begin{figure}[H]
\hfill
\begin{tikzpicture}[
    font=\sffamily\scriptsize,
    line cap=round,
    line join=round
]

\draw[black!60, line width=0.45pt, rounded corners=2.5pt]
    (0,0) rectangle (6.5,1.65);

\draw[black!55, line width=0.4pt]
    (3.2,0.15) -- (3.2,1.50);

\node[anchor=west,font=\sffamily\bfseries\scriptsize]
    at (0.16,1.40) {Permutohedron};

\node[anchor=west,font=\sffamily\scriptsize]
    at (0.16,1.10) {(adjacent transpositions)};

\node[anchor=west,font=\sffamily\bfseries\scriptsize]
    at (3.33,1.40) {Interior edges};

\node[anchor=west,font=\sffamily\scriptsize]
    at (3.33,1.10) {(additional transpositions)};

\draw[red,line width=1.6pt]
    (0.32,0.75) -- (1.95,0.75);
\node[anchor=west] at (2.18,0.75) {(12)};

\draw[blue,line width=1.6pt]
    (0.32,0.47) -- (1.95,0.47);
\node[anchor=west] at (2.18,0.47) {(23)};

\draw[green!70!black,line width=1.6pt]
    (0.32,0.19) -- (1.95,0.19);
\node[anchor=west] at (2.18,0.19) {(34)};

\draw[orange,line width=1.2pt,dashed]
    (3.49,0.75) -- (5.47,0.75);
\node[anchor=west] at (5.70,0.75) {(13)};

\draw[violet,line width=1.2pt,dashed]
    (3.49,0.47) -- (5.47,0.47);
\node[anchor=west] at (5.70,0.47) {(14)};

\draw[cyan!80!black,line width=1.2pt,dashed]
    (3.49,0.19) -- (5.47,0.19);
\node[anchor=west] at (5.70,0.19) {(24)};

\end{tikzpicture}
\end{figure}
\begin{figure}[H]
\centering
\vspace{-\intextsep}
\includegraphics[width=\columnwidth]{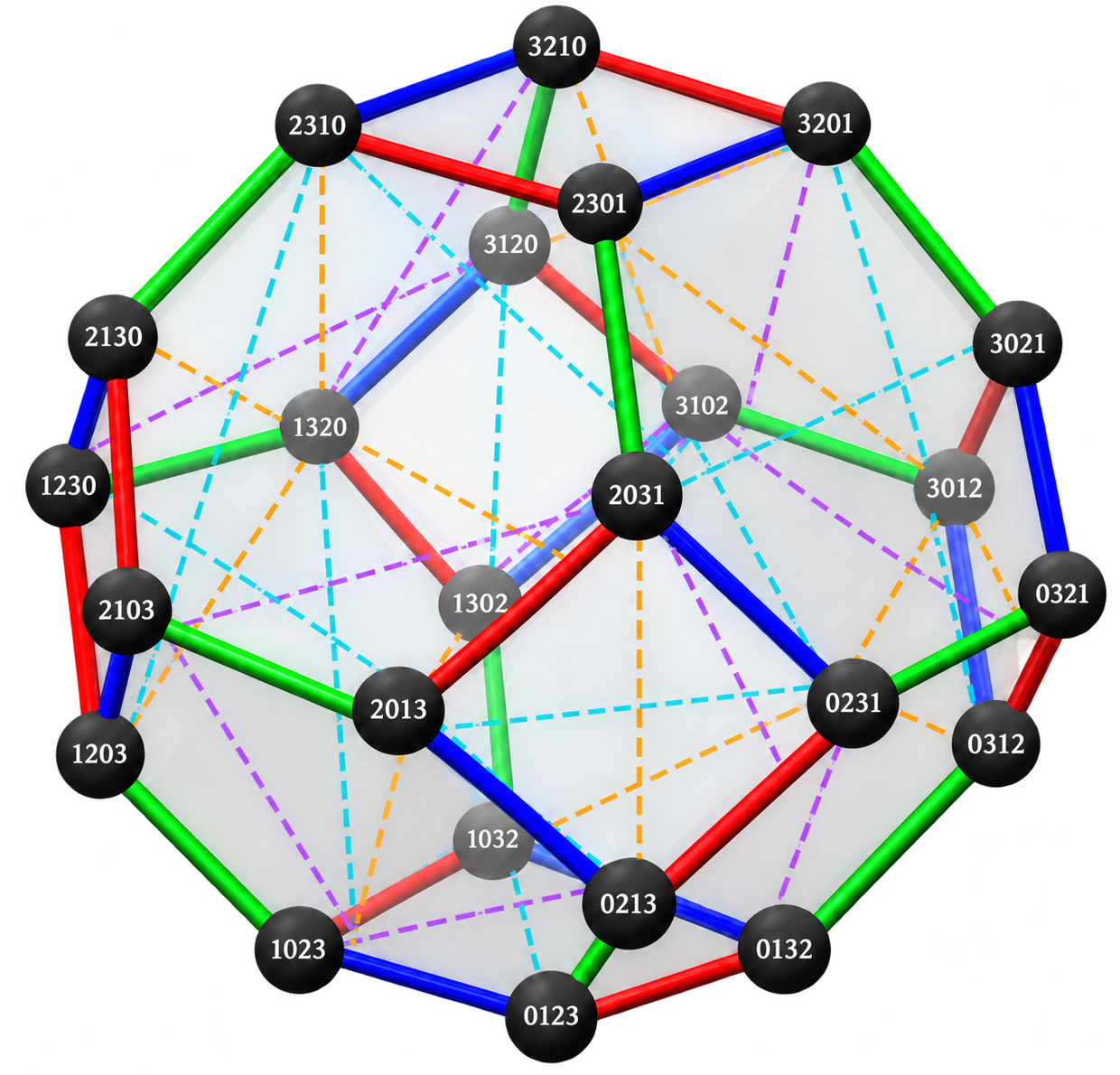}
\caption{
The permutohedral Cayley graph and its conjugacy completion for $S_4$. The solid edges are the adjacent transpositions $(12)$, $(23)$, and $(34)$ and form the permutohedral $1$-skeleton $\operatorname{Cay}(S_4,\{(12),(23),(34)\})$. The dashed edges are $(13)$, $(14)$, and $(24)$. Together these six transpositions form the conjugacy class $C_{(2,1^2)}$ and the normal Cayley graph $\operatorname{Cay}(S_4,C_{(2,1^2)})$. The permutohedral $1$-skeleton is a spanning subgraph of the normal
Cayley graph. The vertical embedding is graded by inversion number, or Coxeter length, from the identity $0123$ at the bottom to the longest element $3210$ at the top; orienting the solid edges upward is the weak Bruhat order on $S_4$.
}
\label{fig:s4-normal-cayley}
\end{figure}

For each  \(C\in\operatorname{Conj}(\mathcal G)\), define the relation
\[
R_C
=
\{(g,h)\in\mathcal G\times\mathcal G:gh^{-1}\in C\}
\]
and its adjacency matrix
\(A_C\in\mathbb C^{\mathcal G\times\mathcal G}\) by
\begin{equation}
(A_C)_{g,h}
=
\mathds{1}_{\{gh^{-1}\in C\}}.
\end{equation}
\(A_C\) is the adjacency matrix of the relation graph associated
with \(C\). Moreover,
$
A_C^{\mathsf T}=A_{C^{-1}},
$
so the relation graph is undirected exactly when \(C=C^{-1}\). After ordering
the conjugacy classes as \(C_1,\ldots,C_d\), let
$
A_i:=A_{C_i},
$ and $
K_i:=K_{C_i}.
$
The collection \(\{A_1,\ldots,A_d\}\) forms the \textit{conjugacy-class association scheme} of \(\mathcal G\), and its span is a commutative Bose--Mesner algebra~\cite{BrouwerCohenNeumaier1989,Godsil1993,ArnadottirGodsil2023}. 
If
\(\rho:\mathcal G\to GL(\mathbb C[\mathcal G])\)
is the left-regular representation, then
$
\rho(K_i)
=
\sum_{g\in C_i}\rho(g)
$
is the left-convolution operator associated with the class sum \(K_i\). In the group-element basis
\(\{\ket g\}_{g\in\mathcal G}\), this operator has matrix entries
\[
(\rho(K_i))_{h,h'}
=
\mathds{1}_{\{hh'^{-1}\in C_i\}}
=
(A_i)_{h,h'}.
\]
Thus the adjacency algebra of the conjugacy-class scheme is isomorphic to
the center of the group algebra,
$
\operatorname{span}\{A_i\}_{i=1}^{d}
\cong
Z(\mathbb C[\mathcal G]).
$
In fact,
\begin{equation}
A_iA_j
=
\sum_{k=1}^{d}p_{ij}^{k}A_k,
\end{equation}
where \(p_{ij}^{k}\) are the \emph{intersection numbers} of the scheme, counting the number of ways a path of types \(i,j\) whose endpoints have relation \(k\). The sums of the classes \(K_i\) are in \(Z(\mathbb C[\mathcal G])\), so the matrices \(A_i\) can also be diagonalized simultaneously.

The matrices \(A_i\) act on the full
\(|\mathcal G|\)-dimensional group-element space \(\mathbb C[\mathcal G]\), and yet preserve \(Z(\mathbb C[\mathcal G])\). Its restriction to the center,
expressed in the normalized conjugacy-class basis, is the class-convolution matrix \(T_C\) defined in Section III.
\[
T_C
=
\left[
\left.A_C\right|_{Z(\mathbb C[\mathcal G])}
\right]_{\{\ket{C_i}\}},
\qquad
(T_C)_{ij}
=
\langle C_i|A_C|C_j\rangle
\]
For \(S_4\), the Cayley-graph adjacency operators \(A_i\) are \(24\times24\), while their central restrictions \(T_C\) are the
\(5\times5\) class-convolution matrices displayed in \eqref{s4-classconv-matrices}.


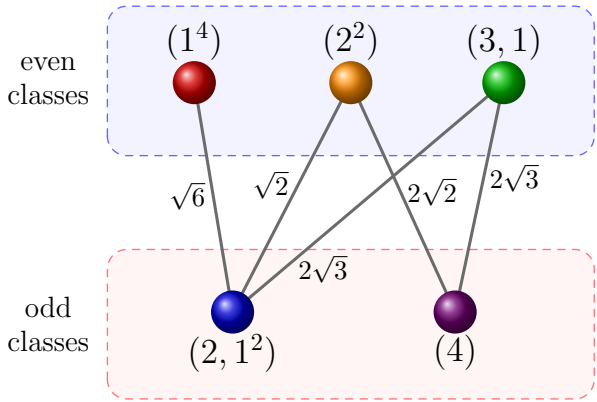
\begin{figure}[H]
\centering
\resizebox{0.95\columnwidth}{!}{%
\begin{tikzpicture}[
    x=1cm,
    y=1cm,
    line cap=round,
    line join=round,
    classlabel/.style={
    font=\fontsize{14pt}{17pt}\selectfont
    },
    weight/.style={
        font=\fontsize{11pt}{10pt}\selectfont,
        inner sep=0pt
    },
    edge/.style={
        draw=black!58,
        line width=1.25pt
    }
]


\fill[blue!7, fill opacity=.65, rounded corners=8pt]
    (0.85,3.80) rectangle (7.85,5.95);
\draw[blue!60, dashed, line width=.55pt, rounded corners=8pt]
    (0.85,3.80) rectangle (7.85,5.95);

\fill[red!6, fill opacity=.65, rounded corners=8pt]
    (0.85,0.30) rectangle (7.85,2.45);
\draw[red!55, dashed, line width=.55pt, rounded corners=8pt]
    (0.85,0.30) rectangle (7.85,2.45);

\node[align=center,font=\large] at (0.,4.88)
    {even\\classes};

\node[align=center,font=\large] at (0.0,1.38)
    {odd\\classes};


\coordinate (id)  at (2.10,4.85);
\coordinate (c22) at (4.35,4.85);
\coordinate (c31) at (6.55,4.85);

\coordinate (c211) at (2.65,1.55);
\coordinate (c4)   at (5.85,1.55);


\draw[edge] (id)   -- (c211);
\draw[edge] (c22)  -- (c211);
\draw[edge] (c31)  -- (c211);
\draw[edge] (c22)  -- (c4);
\draw[edge] (c31)  -- (c4);


\node[weight] at (2,3.22) {$\sqrt{6}$};

\node[weight] at (3.2,3.38) {$\sqrt{2}$};

\node[weight] at (3.95,2.17) {$2\sqrt{3}$};

\node[weight] at (5.52,3.28) {$2\sqrt{2}$};

\node[weight] at (6.7,3.48) {$2\sqrt{3}$};


\shade[ball color=red!85!black]
    (id) circle (0.31);

\shade[ball color=orange!90!yellow]
    (c22) circle (0.31);

\shade[ball color=green!75!black]
    (c31) circle (0.31);

\shade[ball color=blue!85!black]
    (c211) circle (0.31);

\shade[ball color=violet!90!black]
    (c4) circle (0.31);


\node[classlabel,above=6pt of id]
    {$(1^4)$};

\node[classlabel,above=6pt of c22]
    {$(2^2)$};

\node[classlabel,above=6pt of c31]
    {$(3,1)$};

\node[classlabel,below=6pt of c211]
    {$(2,1^2)$};

\node[classlabel,below=6pt of c4]
    {$(4)$};

\end{tikzpicture}%
}
\caption{
Bipartite class graph for the transposition class of $S_4$. The vertices are the normalized conjugacy-class states
$
|C_i\rangle,
$
and edge weights are the nonzero matrix elements of the central restriction $T_{(2,1^2)}$ of $A_{(2,1^2)} = \sum_{\tau\in C_{(2,1^2)}}L_\tau.$ Since multiplication by a transposition reverses permutation parity, $T_{(2,1^2)}$ couples only even and odd classes.
}

\label{fig:s4_restriction}
\end{figure}

A \textit{normal Cayley graph} is obtained by choosing a connection set \(\mathcal S\subseteq\mathcal G\setminus\{e\}\) that is both inverse closed and conjugacy invariant. Equivalently, \(\mathcal S\) is an inversion-closed union of conjugacy classes,
\[
\mathcal S=\bigcup_{i\in \mathcal I}C_i,
\qquad
C_i\subseteq\mathcal S
\Longrightarrow
C_i^{-1}\subseteq\mathcal S.
\]
The graph \(\operatorname{Cay}(\mathcal G,\mathcal S)\) has general adjacency operator
$
A_{\mathcal S}
=
\sum_{i\in I}A_i.
$
Assigning a weight \(w_i\) to each conjugacy class produces the weighted central adjacency
\begin{equation}
A_w
=
\sum_{i=1}^{d}w_iA_i.
\end{equation}
Related work treats general weighted Cayley graphs through irreducible matrix coefficients~\cite{BeckEtAl2024}. A normal Cayley graph specifies one operator in the larger commuting adjacency algebra of the conjugacy-class scheme. The QCT provides a common spectral basis for this entire algebra, simultaneously covering every choice of connection set and class weights.

Association schemes have two canonical bases. The matrices $A_i$ form the
adjacency, or relation, basis, while the primitive idempotent projectors
$P_\Gamma$ form their common spectral basis.

\begin{remark}[Peirce interpretation]
Define the primitive central idempotents of
$\mathbb C[\mathcal G]$
\begin{equation}
P_\Gamma
=
\frac{d_\Gamma}{|\mathcal G|}
\sum_{g\in\mathcal G}
\chi_\Gamma(g)\,\rho(g).
\
\end{equation}
$P_\Gamma$ projects onto the dual $\Gamma^*$-isotypic component, and satisfies
$
P_\Gamma^\dagger=P_\Gamma,
$
\[
P_\Gamma P_\mu
=
\delta_{\Gamma\mu}P_\Gamma,
\qquad
\sum_{\Gamma\in\widehat{\mathcal G}}P_\Gamma=
I_{|\mathcal G|},
\]
This is a \textit{central Peirce decomposition}
\[
\operatorname{span}\{A_i\}_{i=1}^{d}
=
\bigoplus_{\Gamma\in\widehat{\mathcal G}}
\mathbb C P_\Gamma.
\]
\end{remark}
Writing
$
A_i
=
\sum_\Gamma p_i(\Gamma)P_\Gamma,
$
the coefficients $p_i(\Gamma)$ are the eigenvalues of $A_i$ on the range of $P_\Gamma$. By Schur's lemma,
\begin{equation}
p_i(\Gamma)
=
\frac{|C_i|\,\chi_\Gamma(C_i^{-1})}{d_\Gamma}
=
\frac{|C_i|\,\overline{\chi_\Gamma(C_i)}}{d_\Gamma}.
\label{eigenp}
\end{equation}
In the regular representation, each projector has rank $d_\Gamma^2$. For an individual operator, distinct sectors may share the same eigenvalue.

The matrix
$
\mathsf P_{\mathrm{eig}}
=
\bigl(p_i(\Gamma)\bigr)_{\Gamma,i}
$
is the first eigenmatrix of the conjugacy-class association scheme. The standard eigenmatrix orthogonality relations supply the class-size and multiplicity weights that render the character transform
unitary~\cite{BrouwerCohenNeumaier1989,ArnadottirGodsil2023}. After normalization, the change between the adjacency and primitive idempotent bases is expressed by the QCT.

Recent work on association schemes also provides useful context for future direction. Translation association schemes make the connection between convolution, group characters, and graph eigenmatrices, while quantum association schemes extend the formalism to noncommutative finite quantum spaces and quantum graphs~\cite{gromada2026quantumassociationschemes}.  These coherent algebras carry both composition and a generalized Schur product. Under composition, each is a finite-dimensional \(C^*\)-algebra whose center admits a QCT-type transform: an orthonormal central basis and the primitive central projections, normalized using a faithful trace, are related by a unitary transform that diagonalizes multiplication by central operators. The center need not be closed under the generalized Schur product and therefore need not itself define a quantum association scheme.

Returning to the conjugacy-class association scheme, its algebra contains more than adjacency matrices. Polynomials and functions of central adjacency operators remain diagonal in the common character basis. Their restrictions to class space can therefore be represented and processed in the smaller \(d\)-dimensional basis. These operators include graph Laplacians, diffusion operators, random-walk kernels, and Hamiltonians generating continuous-time quantum walks. The following subsections develop these applications.

\subsection{Laplacians, Diffusion, and Graph Filtering}

Let
$
A_w=\sum_{i\in I}w_iA_i
$
be the weighted adjacency operator of a normal Cayley graph of $\mathcal G$, where the weights are real and nonnegative and satisfy $w_i=w_{i^*}$ whenever $C_{i^*}=C_i^{-1}$. The graph is weighted regular with degree
$
\kappa
=
\sum_{i\in \mathcal I}w_i|C_i|.
$
For $\kappa>0$, its combinatorial and normalized Laplacians are
\begin{equation}
\mathscr L
=
D-A_w
=
\kappa I-A_w,
\qquad
\mathscr L_{\rm norm}
=
I-\frac{1}{\kappa}A_w.
\end{equation}

Writing \(\theta_\Gamma\) for the adjacency eigenvalues obtained above, the
Laplacian eigenvalues in sector \(\Gamma\) are
\[
\mu_\Gamma
=
\kappa-\theta_\Gamma,
\qquad
\nu_\Gamma
=
1-\frac{\theta_\Gamma}{\kappa},
\]
for the combinatorial and normalized Laplacians, respectively. The trivial irrep, for which
$
\chi_{\rm triv}(g)=1
$
for every $g\in\mathcal G$, gives
$
\theta_{\rm triv}=\kappa
$
and
$
\mu_{\rm triv}=0,
$
with eigenvector the constant vector. This zero eigenvalue is unique when the graph is connected.

On the central subspace, let
$
T_w=\left.A_w\right|_{Z(\mathbb C[\mathcal G])}
$
and
$
\mathscr L_{\rm cl}=\kappa I-T_w.
$
The QCT diagonalizes these compressed operators:
\[
\begin{aligned}
U_{\rm QCT}\mathscr L_{\rm cl}U_{\rm QCT}^\dagger
&=
\operatorname{diag}(\mu_\Gamma),\\
U_{\rm QCT}
\left(
I-\frac{1}{\kappa}T_w
\right)
U_{\rm QCT}^\dagger
&=
\operatorname{diag}(\nu_\Gamma).
\end{aligned}
\]

For a general undirected graph $\mathcal X$ with Laplacian $\mathscr L$,
graph Fourier analysis chooses an orthonormal eigenbasis $\{u_j\}$ satisfying
$
\mathscr Lu_j=\mu_j u_j
$
and defines the graph Fourier coefficients of a signal $f$ by
$
\widehat f_j=\langle u_j,f\rangle.
$
In matrix form, if
$
\mathscr L=U\Gamma U^\dagger,
$
the graph Fourier transform is
$
\mathscr F_{\mathcal X}=U^\dagger.
$
Unlike the ordinary Fourier transform, $U$ depends on the particular graph and must be obtained from its spectral decomposition~\cite{ShumanEtAl2013}. For conjugacy-invariant Cayley graphs, the QCT basis is fixed by the representation theory of $\mathcal G$ and does not depend on the class weights.

Functional calculus preserves the same spectral projectors. On the full
group-element space,
\[
e^{-t\mathscr L}
=
\sum_{\Gamma\in\widehat{\mathcal G}}
e^{-t\mu_\Gamma}P_\Gamma,
\qquad
e^{-it\mathscr L}
=
\sum_{\Gamma\in\widehat{\mathcal G}}
e^{-it\mu_\Gamma}P_\Gamma.
\]
The first operator is the heat semigroup, or continuous-time random-walk evolution, while the second is continuous-time quantum-walk evolution with the Laplacian as its Hamiltonian~\cite{FarhiGutmann1998}. Restricting these operators to the central subspace,
\[
\begin{aligned}
e^{-t\mathscr L_{\rm cl}}
&=
U_{\rm QCT}^\dagger
\operatorname{diag}(e^{-t\mu_\Gamma})
U_{\rm QCT},\\
e^{-it\mathscr L_{\rm cl}}
&=
U_{\rm QCT}^\dagger
\operatorname{diag}(e^{-it\mu_\Gamma})
U_{\rm QCT}.
\end{aligned}
\]

A graph filter $\varphi(\mathscr L_{\rm cl})$, such as the resolvent low-pass filter
$
(I+\alpha\mathscr L_{\rm cl})^{-1}
$
or a band-pass filter, also acts diagonally in the QCT basis,
\[
\varphi(\mathscr L_{\rm cl})
=
U_{\rm QCT}^\dagger
\operatorname{diag}\bigl(\varphi(\mu_\Gamma)\bigr)
U_{\rm QCT}.
\]
The irreducible-character sectors $\Gamma$ define the graph-frequency
decomposition, ordered by the Laplacian eigenvalues $\mu_\Gamma$, with
larger values indicating higher graph frequencies~\cite{ShumanEtAl2013,HammondEtAl2011}.

\subsection{Random and Quantum Walks}

Let
$
p:\mathcal G\to[0,1]
$
be a conjugacy-invariant probability distribution, and write
$
p_i=p(s)
$
for $s\in C_i$. Its normalization is
$
\sum_{i=1}^{d}|C_i|p_i=1.
$
Using the class-sum adjacency operators defined above, the transition operator
of the discrete-time random walk is
\begin{equation}
\mathsf P
=
\sum_{i=1}^{d}p_iA_i.
\end{equation}
Viewing \(p\) as a class function, define
$
\theta_\Gamma
:=
\widehat p(\Gamma),
$
where \(\widehat p(\Gamma)\) is given by Eq.~\eqref{eigenp}. Let
$
\mathsf P_{\rm cl}
=
\left.\mathsf P\right|_{\mathrm{Cl}(\mathcal G)}
$
denote the induced action on the class-function space. The discrete character transform exposes the scalar eigenvalues:
\begin{equation}
U_{\rm QCT}\mathsf P_{\rm cl}U_{\rm QCT}^{\dagger}
=
\operatorname{diag}
\left(
\theta_\Gamma
\right)_{\Gamma\in\widehat{\mathcal G}}.
\label{eq:qct-random-walk-diagonalization}
\end{equation}
A central random walk on the full group algebra is compressed by the QCT to
one spectral coordinate for each irreducible character.

For a symmetric ergodic walk, the trivial representation satisfies
$
\theta_{\rm triv}=1,
$
while the remaining eigenvalues determine convergence to stationarity. The
absolute spectral gap is
$
\gamma
=
1-
\max_{\Gamma\neq{\rm triv}}
|\theta_\Gamma|.
$
The mixing behavior of a conjugacy-invariant walk can be expressed in terms of these character ratios.

A standard example is the random-transposition walk on $S_n$. Let
$C_{\rm tr}$ denote the conjugacy class of transpositions, with
$
|C_{\rm tr}|=\binom{n}{2},
$
and choose each transposition $s\in C_{\rm tr}$ with probability
$
p(s)=\binom{n}{2}^{-1}.
$
Equation~\eqref{eigenp} becomes
\[
\theta_\Gamma
=
\frac{\chi_\Gamma((12))}{d_\Gamma}.
\]
The spectrum of the walk is encoded by the normalized characters of the transposition class, as in the classical representation-theoretic analysis of random walks on finite groups~\cite{Diaconis1988}. The non-lazy random-transposition walk is periodic because every step changes permutation parity; in particular, the sign representation has eigenvalue $-1$. A lazy version of the walk is required for convergence to the uniform distribution in total variation.

Repeated evolution remains diagonal in character coordinates, which is the character-theoretic version of diagonalizing a translation-invariant walk by the ordinary Fourier transform. Convolution becomes multiplication, and the $t$-step evolution is obtained by raising each character-sector eigenvalue to the $t$th power. The QCT records one sector eigenvalue \(\theta_\Gamma\) for each
\(\Gamma\in\widehat{\mathcal G}\), whereas on the regular
representation \(\mathbb C[\mathcal G]\), this eigenvalue has
multiplicity \(d_\Gamma^2\). Therefore,
\begin{equation}
\operatorname{Tr}(\mathsf P^t)
=
\sum_{\Gamma\in\widehat{\mathcal G}}
d_\Gamma^2\theta_\Gamma^t.
\end{equation}
Translation invariance makes the diagonal entries of \(\mathsf P^t\)
equal, so the \(t\)-step return probability to any
\(g\in\mathcal G\) is
\begin{equation}
(\mathsf P^t)_{g,g}
=
\frac{1}{|\mathcal G|}
\sum_{\Gamma\in\widehat{\mathcal G}}
d_\Gamma^2\theta_\Gamma^t.
\label{eq:return-probability}
\end{equation}

The same spectral description applies to continuous-time quantum walks. Let
$
H
=
\sum_{i=1}^{d}w_iA_i
$
be a central Hermitian Hamiltonian, with class weights satisfying
$
w_{i^*}=\overline{w_i}
$
whenever
$
C_{i^*}=C_i^{-1}.
$
Let $w\in\mathrm{Cl}(\mathcal G)$ be the class function defined by
$
w(g)=w_i
$
for $g\in C_i$. The Hamiltonian eigenvalues are
$
\eta_\Gamma
:=
\widehat w(\Gamma),
$
where \(\widehat w(\Gamma)\) is given by Eq.~\eqref{eigenp}. Its
compressed action
$
H_{\rm cl}
=
\left.H\right|_{\mathrm{Cl}(\mathcal G)}
$
has the evolution
\begin{equation}
e^{-itH_{\rm cl}}
=
U_{\rm QCT}^{\dagger}
\operatorname{diag}
\left(
e^{-it\eta_\Gamma}
\right)_{\Gamma\in\widehat{\mathcal G}}
U_{\rm QCT}.
\label{eq:qct-quantum-walk-evolution}
\end{equation}
Evolution under a central Hamiltonian reduces in QCT coordinates to
independent phase evolution of the irreducible-character
sectors~\cite{FarhiGutmann1998}.

The diagonal form also allows a direct measurement interface for spectral information. For a normalized class-function state
$
|\psi\rangle
=
\sum_{i=1}^{d}a_i|C_i\rangle,
$
\begin{equation}
U_{\rm QCT}|\psi\rangle
=
\sum_{\Gamma\in\widehat{\mathcal G}}
\widehat a_\Gamma|\Gamma\rangle.
\label{qct-coords}
\end{equation}
Measurement in the character basis samples $\Gamma$ according to
$
\Pr(\Gamma\mid\psi)
=
|\widehat a_\Gamma|^2.
$
The character-sector weights can be sampled without reconstructing the full
transformed state classically, and observables diagonal in the character
basis can be estimated from repeated measurements.

For a real-valued function $\varphi $ on the spectrum of $H_{\rm cl}$,
\begin{equation}
\varphi (H_{\rm cl})
=
U_{\rm QCT}^{\dagger}
\left(
\sum_{\Gamma\in\widehat{\mathcal G}}
\varphi (\eta_\Gamma)
|\Gamma\rangle\langle\Gamma|
\right)
U_{\rm QCT}.
\end{equation}
Its expectation in the state $|\psi\rangle$ is
\begin{equation}
\langle \varphi (H_{\rm cl})\rangle_\psi
=
\sum_{\Gamma\in\widehat{\mathcal G}}
|\widehat a_\Gamma|^2\varphi (\eta_\Gamma).
\label{eq:qct-spectral-expectation}
\end{equation}
$
\varphi (x)=x^k
$
gives the spectral moments
$
\langle H_{\rm cl}^k\rangle_\psi
=
\sum_{\Gamma\in\widehat{\mathcal G}}
|\widehat a_\Gamma|^2\eta_\Gamma^k.
$

For a subset
$
S\subseteq\widehat{\mathcal G},
$
the projector onto the sectors indexed by $S$ is
\begin{equation}
\Pi_S
=
U_{\rm QCT}^{\dagger}
\left(
\sum_{\Gamma\in S}
|\Gamma\rangle\langle\Gamma|
\right)
U_{\rm QCT},
\end{equation}
with
$
\langle\psi|\Pi_S|\psi\rangle
=
\sum_{\Gamma\in S}
|\widehat a_\Gamma|^2.
$
This quantity measures the total spectral weight carried by a chosen family of symmetry sectors.

For quantum walks, the first moments determine average spectral quantities, higher moments probe the eigenvalue distribution seen by the initial state, and sector projectors determine its overlap with prescribed symmetry
components. The phases $e^{-it\eta_\Gamma}$ control periodicity, interference, time-averaged mixing, and conditions for state transfer. The QCT acts both as a diagonalizing transform for class-invariant walks and as an interface for sampling their symmetry-adapted spectral observables.

\subsection{Class-Symmetric Hamiltonians and Exchange Dynamics}

Let
$
\rho:\mathcal G\to U(\mathcal H)
$
be a unitary representation on a physical Hilbert space $\mathcal H$. A
\emph{class-symmetric} Hamiltonian is one of the form
\begin{equation}
H
=
\sum_{g\in\mathcal G}J_g\,\rho(g),
\end{equation}
where
$
J_{hgh^{-1}}=J_g
$
and
$
J_{g^{-1}}=\overline{J_g}.
$
The first condition makes the coupling constant on each conjugacy class,
while the second ensures that $H=H^\dagger$. In terms of the class sums
$
K_i,
$
the Hamiltonian is
$
H
=
\sum_{i=1}^{d}J_i\,\rho(K_i),
$
where $J_i=J_g$ for $g\in C_i$.

Each $\rho(K_i)$ commutes with $\rho(\mathcal G)$. If the representation decomposes as
\[
\mathcal H
\cong
\bigoplus_{\Gamma\in\widehat{\mathcal G}}
\mathbb C^{m_\Gamma}\otimes V_\Gamma,
\]
then by Schur's lemma
\[
\left.\rho(K_i)\right|_{\mathbb C^{m_\Gamma}\otimes V_\Gamma}
=
I_{m_\Gamma}
\otimes
\frac{|C_i|\chi_\Gamma(C_i)}{d_\Gamma}
I_{V_\Gamma}.
\]
The energy associated with the irreducible sector $\Gamma$ is
\begin{equation}
E_\Gamma
=
\sum_{i=1}^{d}
J_i\,
\frac{|C_i|\chi_\Gamma(C_i)}{d_\Gamma}.
\end{equation}
The QCT diagonalizes the class-algebra representation of these operators, while their action on $\mathcal H$ follows from the
irreducible decomposition of $\rho$.

A paradigmatic example is the uniform permutation-exchange interaction of
$n$ identical spin-$\frac12$ particles. Let $\pi_{ij}:=\rho((ij))$ denote the operator that swaps spins $i$ and $j$. The group $S_n$ acts on the $n$-spin space by permuting tensor factors, and
the Hamiltonian
\begin{equation}
H_{\rm ex}
=
J\sum_{1\leq i<j\leq n}\pi_{ij}
=
J\,\rho\!\left(K_{(2,1^{n-2})}\right)
\end{equation}
is the class sum of all transpositions in $S_n$. Its spectrum on the
$n$-spin Hilbert space is
$
E_\Gamma
=
J\binom{n}{2}
\frac{\chi_\Gamma((12))}{d_\Gamma},
$
where only partitions $\lambda\vdash n$ with at most two rows occur in
$(\mathbb C^2)^{\otimes n}$ by Schur--Weyl duality~\cite{FultonHarris1991}. This recovers the energy levels of the
permutation-symmetric exchange model without directly diagonalizing its
$2^n\times2^n$ Hamiltonian. The trivial representation has energy
$
E_{\rm triv}=J\binom{n}{2};
$
whether this is the highest or lowest energy depends on the sign of $J$. For even $n$, the total-spin-zero sector is the sector labeled by the partition $(n/2,n/2)$, not the sign representation, except when $n=2$. On the other hand, a nearest-neighbor chain
$
H_{\rm nn}
=
J\sum_{i=1}^{n-1}\pi_{i,i+1}
$
is not a class sum, since the adjacent transpositions form only part of the
full transposition class. The set of terms is not invariant under conjugation
by $S_n$, so $H_{\rm nn}$ does not commute with the full permutation action
and the character values alone do not diagonalize it.

Returning to $S_4$, the most general central class-sum Hamiltonian on the
24-dimensional regular space is
\begin{align}
H
={}&
J_0I
+J_2\rho\!\left(K_{(2,1^2)}\right)
+J_{22}\rho\!\left(K_{(2^2)}\right)
\nonumber\\
&+
J_3\rho\!\left(K_{(3,1)}\right)
+J_4\rho\!\left(K_{(4)}\right)
\end{align}
where $J_0,J_2,J_{22},J_3,J_4\in\mathbb R$. Modulo the overall energy shift $J_0I$, this is a four-parameter family.
Restricting to the five-dimensional center and applying the QCT gives the diagonal matrix $\operatorname{diag}(E_\Gamma)$, where
\begin{align}
E_\Gamma
={}&
J_0
+J_2\frac{6\,\chi_\Gamma((12))}{d_\Gamma}
+J_{22}\frac{3\,\chi_\Gamma((12)(34))}{d_\Gamma}
\nonumber\\
&+
J_3\frac{8\,\chi_\Gamma((123))}{d_\Gamma}
+J_4\frac{6\,\chi_\Gamma((1234))}{d_\Gamma}.
\end{align}
For the trivial representation,
$
\chi_{\rm triv}(g)=1
$
for all $g\in S_4$, giving
$
E_{\rm triv}
=
J_0+6J_2+3J_{22}+8J_3+6J_4.
$
The remaining four sector energies follow from the $S_4$ character table.

\subsection{Molecular and Permutation--Inversion Dynamics}

Class-sum Hamiltonians also appear in molecular physics. Suppose that a finite permutation--inversion group $\mathcal G$ acts regularly on a set of equivalent molecular configurations, which may then be labeled by
$g\in\mathcal G$. If tunneling or rearrangement processes are represented by elements $s\in\mathcal G$, an effective Hamiltonian can take the form
\begin{equation}
H
=
E_0I+\sum_{s\in\mathcal G}t_s\,\rho(s),   
\end{equation}
where $\rho(s)$ permutes the configurations and
$
t_{s^{-1}}=\overline{t_s}.
$
The configuration space need not be the group itself; it can be a homogeneous $\mathcal G$-set.

If the amplitudes are constant on conjugacy classes, so that
$
t_{hsh^{-1}}=t_s
$
for all $s,h\in\mathcal G$, then
\[
H
=
E_0I+\sum_{i=1}^{d}t_i\,\rho(K_i)
\]
is a class-symmetric Hamiltonian. In the
language of permutation--inversion groups used in molecular spectroscopy~\cite{Tinkham2003,BunkerJensen1998}, this condition assigns
equal amplitudes to rearrangements related by conjugation. The energy in each irreducible sector is obtained from the character formula of the preceding subsection. Such class-invariant models provide a compact description of tunneling splittings associated with large-amplitude molecular motions.

\begin{figure}[H]
\centering
\resizebox{0.96\columnwidth}{!}{%
\begin{tikzpicture}[
    >=Stealth,
    line cap=round,
    line join=round,
    transform/.style={
        ->,
        line width=1.15pt,
        draw=black!75
    },
    oplabel/.style={
        font=\large
    },
    statelabel/.style={
        font=\Large
    }
]


\coordinate (X)  at (0,4.05);
\coordinate (SX) at (5.30,4.05);
\coordinate (HX) at (0,-0.65);
\coordinate (HS) at (5.30,-0.65);

\node[inner sep=0pt] at (X)
    {\includegraphics[width=2.5cm]{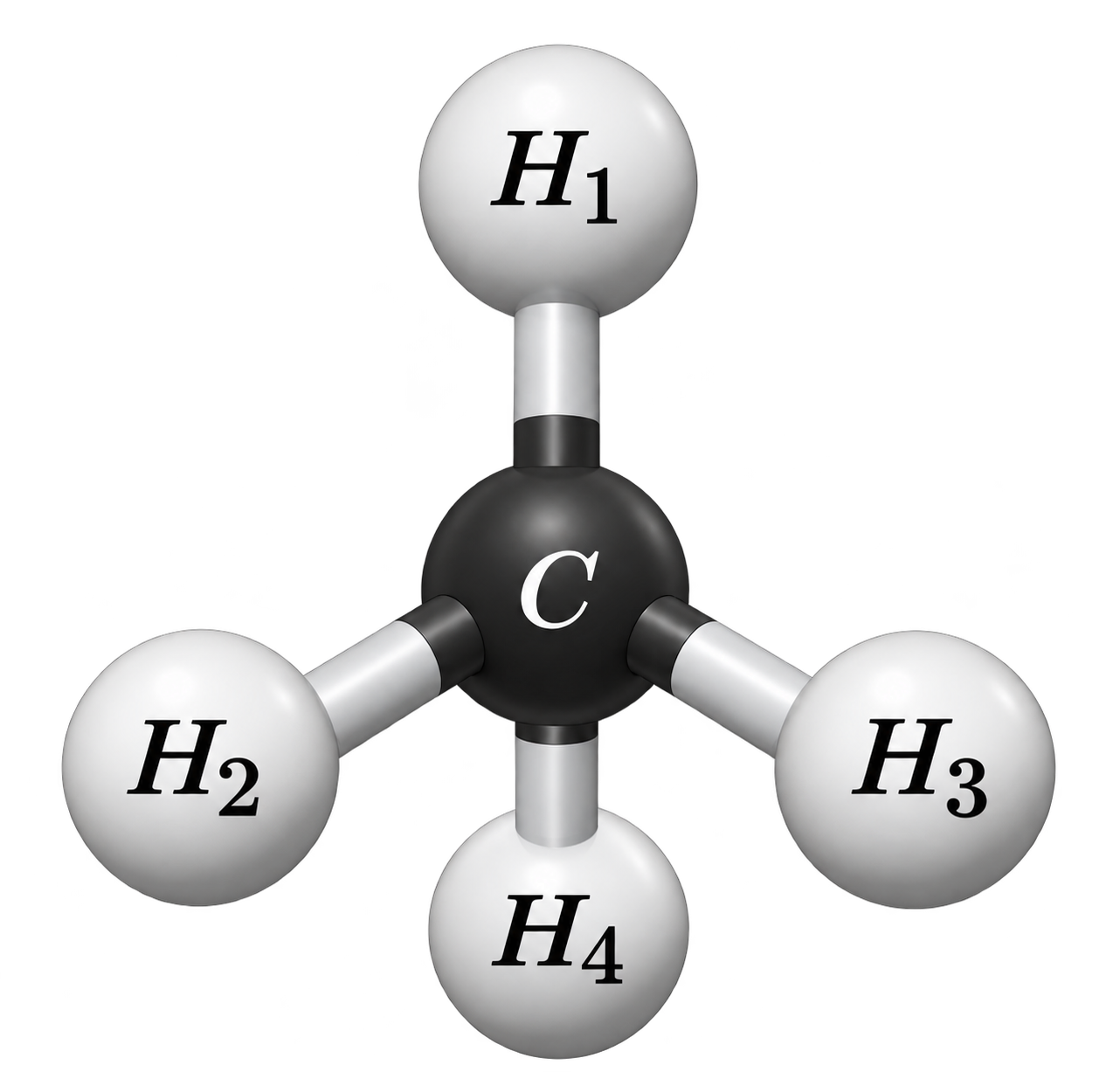}};

\node[inner sep=0pt] at (SX)
    {\includegraphics[width=2.5cm]{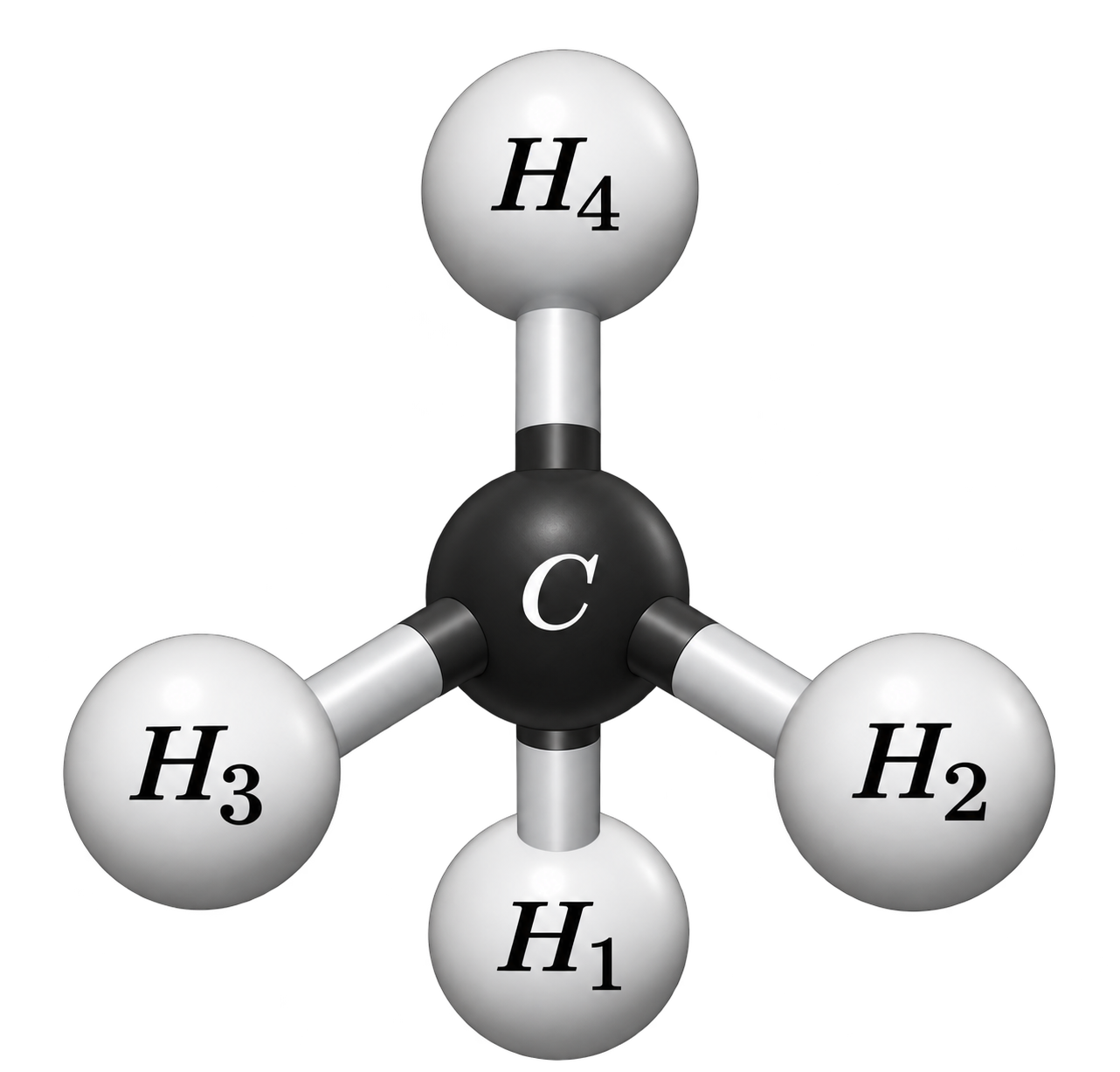}};

\node[inner sep=0pt] at (HX)
    {\includegraphics[width=2.5cm]{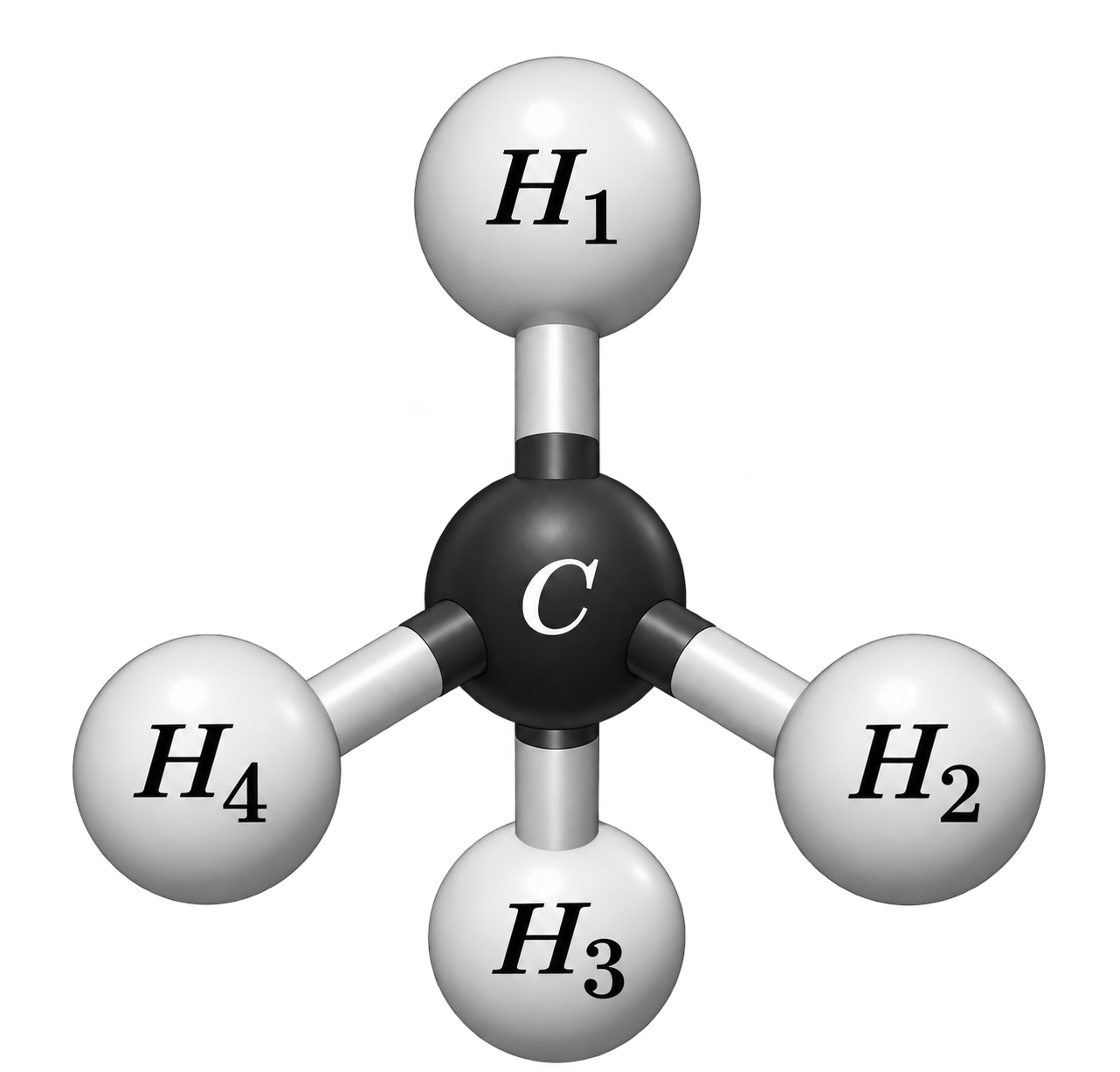}};

\node[inner sep=0pt] at (HS)
    {\includegraphics[width=2.5cm]{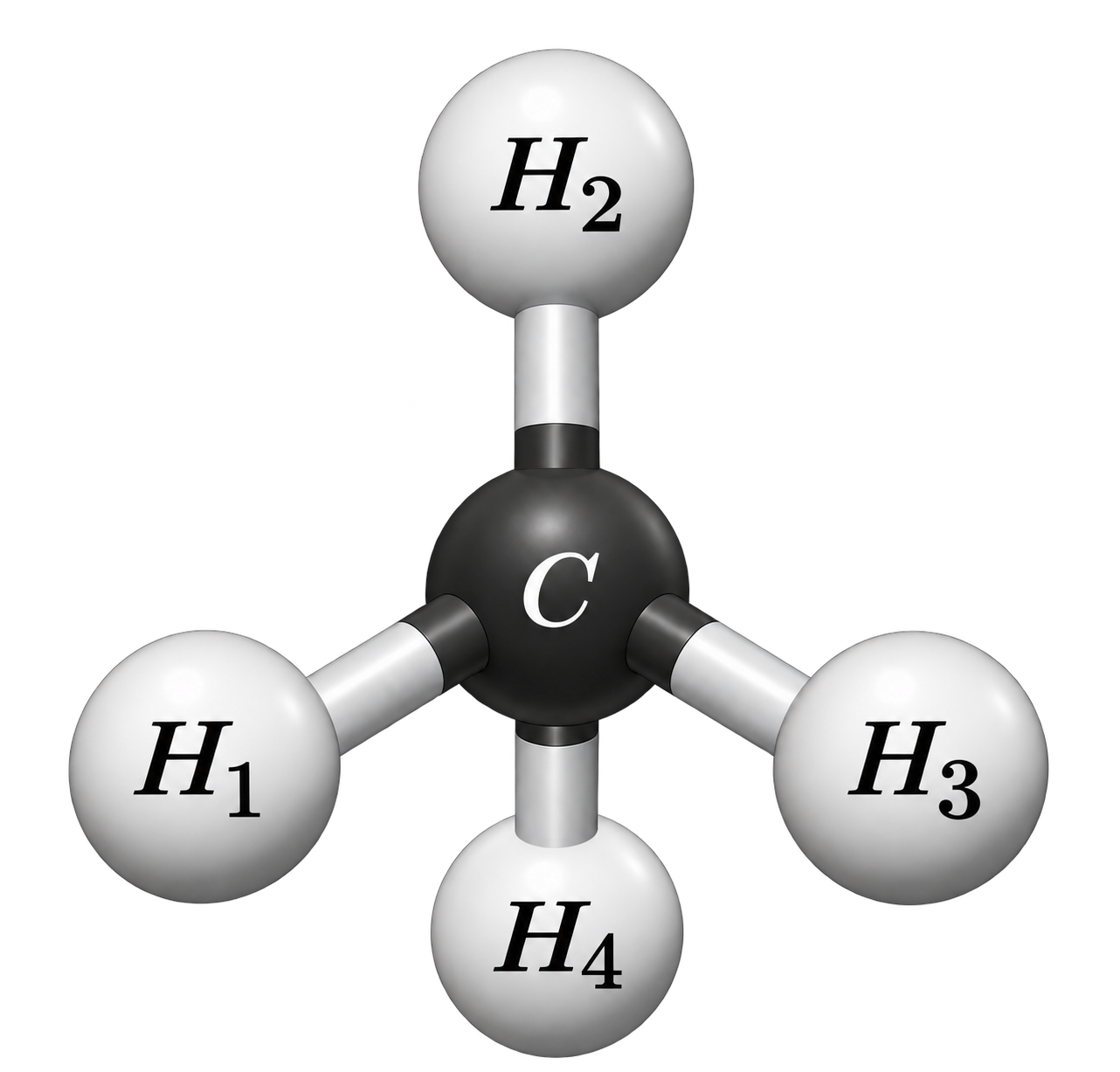}};


\node[statelabel]
    at ($(X)+(0,-1.45)$)
    {$x$};

\node[statelabel]
    at ($(SX)+(0,-1.45)$)
    {$s\!\cdot\!x$};

\node[statelabel]
    at ($(HX)+(0,-1.45)$)
    {$h\!\cdot\!x$};

\node[statelabel]
    at ($(HS)+(0,-1.45)$)
    {$hs\!\cdot\!x$};


\draw[transform]
    ($(X)+(1.38,0)$) --
    ($(SX)+(-1.38,0)$);

\node[oplabel]
    at (2.65,4.42)
    {$s=C_2(z)$};


\begin{scope}[shift={(2.65,5.15)}]

    \draw[
        blue!75!black,
        line width=1.25pt,
        -{Stealth[length=3.5pt,width=3pt]}
    ]
        (0,-0.38) -- (0,0.42);

    \draw[
        blue!75!black,
        line width=1.15pt,
        -{Stealth[length=3.5pt,width=3pt]}
    ]
        (0.36,0.02)
        arc[
            start angle=0,
            end angle=-305,
            x radius=0.36,
            y radius=0.15
        ];

\end{scope}

\node[
    font=\normalsize,
    blue!75!black
]
    at (2.65,5.85)
    {$180^\circ$ about $z$};


\draw[transform]
    ($(X)+(0,-1.72)$) --
    ($(HX)+(0,1.38)$);


\node[oplabel]
    at (-1.65,2.10)
    {$h=C_3([111])$};


\begin{scope}[shift={(-1.65,1.42)},rotate=-54.7356]

    \draw[
        green!55!black,
        line width=1.25pt,
        -{Stealth[length=3.5pt,width=3pt]}
    ]
        (0,-0.38) -- (0,0.42);

    \draw[
        green!55!black,
        line width=1.15pt,
        -{Stealth[length=3.5pt,width=3pt]}
    ]
        (0.36,0.02)
        arc[
            start angle=0,
            end angle=-305,
            x radius=0.36,
            y radius=0.15
        ];

\end{scope}

\node[
    font=\normalsize,
    green!55!black
]
    at (-1.65,0.74)
    {$120^\circ$ about $[111]$};


\draw[transform]
    ($(SX)+(0,-1.72)$) --
    ($(HS)+(0,1.38)$);


\node[oplabel]
    at (6.95,2.10)
    {$h=C_3([111])$};


\begin{scope}[shift={(6.95,1.42)},rotate=-54.7356]

    \draw[
        green!55!black,
        line width=1.25pt,
        -{Stealth[length=3.5pt,width=3pt]}
    ]
        (0,-0.38) -- (0,0.42);

    \draw[
        green!55!black,
        line width=1.15pt,
        -{Stealth[length=3.5pt,width=3pt]}
    ]
        (0.36,0.02)
        arc[
            start angle=0,
            end angle=-305,
            x radius=0.36,
            y radius=0.15
        ];

\end{scope}

\node[
    font=\normalsize,
    green!55!black
]
    at (6.95,0.74)
    {$120^\circ$ about $[111]$};


\draw[transform]
    ($(HX)+(1.38,0)$) --
    ($(HS)+(-1.38,0)$);

\node[oplabel]
    at (2.65,-1.03)
    {$hsh^{-1}=C_2(x)$};


\begin{scope}[shift={(2.65,-1.65)},rotate=-90]

    \draw[
        red!75!black,
        line width=1.25pt,
        -{Stealth[length=3.5pt,width=3pt]}
    ]
        (0,-0.38) -- (0,0.42);

    \draw[
        red!75!black,
        line width=1.15pt,
        -{Stealth[length=3.5pt,width=3pt]}
    ]
        (0.36,0.02)
        arc[
            start angle=0,
            end angle=-305,
            x radius=0.36,
            y radius=0.15
        ];

\end{scope}

\node[
    font=\normalsize,
    red!75!black
]
    at (2.65,-2.30)
    {$180^\circ$ about $x$};

\end{tikzpicture}%
}

\caption{
Conjugate symmetry operations on ideal tetrahedral methane. The tetrahedral symmetry group $T_d\cong S_4$ acts by permuting the four hydrogen vertices. Taking
$H_1=(1,1,1)$,
$H_2=(1,-1,-1)$,
$H_3=(-1,1,-1)$, and
$H_4=(-1,-1,1)$,
the rotation $s=C_2(z)$ induces $(H_1\,H_4)(H_2\,H_3)\in C_{(2^2)}$, while $h=C_3([111])$ induces $(H_2\,H_3\,H_4)\in C_{(3,1)}$. Conjugation gives $hsh^{-1}=C_2(x)$, inducing $(H_1\,H_2)(H_3\,H_4)\in C_{(2^2)}$. The diagram commutes, with
$
h(s\!\cdot\!x)
=
(hsh^{-1})\!\cdot\!(h\!\cdot\!x)
=
hs\!\cdot\!x.
$
}
\label{fig:methane-conjugation}
\end{figure}
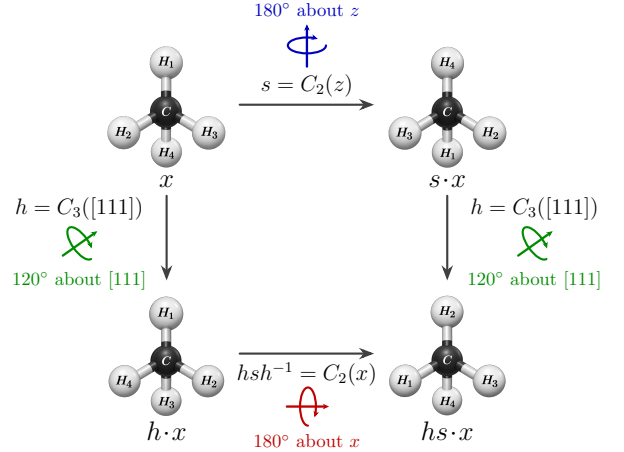

Consider six equivalent configurations arranged in a
cycle and let $r$ generate the cyclic shift between them. If hopping in either direction has amplitude $t$, the Hermitian tunneling Hamiltonian is
\[
H_{\rm hop}
=
E_0I+t\bigl(\rho(r)+\rho(r^{-1})\bigr).
\]
For the cyclic group, the QCT is the ordinary finite Fourier transform and yields the energies
\[
E_k
=
E_0+2t\cos\!\left(\frac{2\pi k}{6}\right),
\qquad
k=0,\ldots,5.
\]
If the model is instead formulated using a dihedral permutation--inversion group, the pair $\{r,r^{-1}\}$ belongs to one conjugacy class, and the same Hamiltonian is a class sum.

\subsection{Symmetry-Aware Learning}
Symmetry-aware learning architectures seek to encode known invariances directly into the representation and admissible learned operators without relying solely on their recovery from augmented data ~\cite{CohenWelling2016,BronsteinEtAl2021}.
For data represented by class functions
$
f\in\mathrm{Cl}(\mathcal G),
$
conjugation invariance reduces the ambient function space on $\mathcal G$ to the $d$-dimensional class space, where
$
d
=
|\operatorname{Conj}(\mathcal G)|
=
|\widehat{\mathcal G}|.
$
The character transform coordinates $\widehat a_\Gamma$ in Eq.~\eqref{qct-coords} give a fixed symmetry-adapted representation that can precede either a classical or quantum learning model.

The character transform also provides a spectral parameterization of the entire algebra of class-convolution layers. A trainable central convolutional layer can be parameterized by its response on the irreducible-character sectors. Such a layer has $d$ spectral parameters, compared with $d^2$ entries for an unrestricted linear operator on the class space. When a class-convolution operator serves as a graph shift, repeated linear propagation multiplies each character-sector amplitude by successive powers of its spectral multiplier. This describes the linear propagation step in message-passing architectures~\cite{DefferrardEtAl2016,KipfWelling2017}.

The same representation permits symmetry-adapted spectral filtering and dimensionality reduction. Writing $f$ in normalized class coordinates, let $\Gamma\subseteq\widehat{\mathcal G}$ and define the diagonal selector
$
D_\Gamma
=
\sum_{\Gamma\in\Gamma}
\ket{\Gamma}\!\bra{\Gamma}.
$
The map
$
f\longmapsto D_\Gamma U_{\mathrm{QCT}}f
$
retains only selected irreducible-character sectors, giving a $|\Gamma|$-dimensional representation after the discarded coordinates are removed. When the class-convolution operators describe adjacency or propagation, this restriction can be viewed as symmetry-adapted spectral pooling, where information is retained or discarded at the level of representation sectors, not individual group elements. Multiple learned class-convolution filters can be combined with sector selection, nonlinear processing, or measurement for classification and discrimination tasks.

\subsection{Conjugacy Orbits and Geometric Invariance}

A geometric interpretation is found when \(\mathcal G\) is represented in \(\operatorname{Pin}(V,Q)\) or \(\operatorname{Spin}(V,Q)\), with \(Q\) nondegenerate. Group conjugation is the Clifford sandwich action,
\begin{equation}
R
\longmapsto
URU^{-1},
\qquad
U,R\in
\mathcal G
\subset
\mathcal{C}\ell(V,Q)^\times,
\end{equation}
so conjugacy classes are orbits of Clifford elements under this action. For a rotor \(R=\exp(-B/2)\) with bivector \(B\in\bigwedge^2 V\), conjugation is
\[
URU^{-1}
=
\exp\!\left(
-\frac{UBU^{-1}}{2}
\right).
\]
Conjugation transports the generating bivector by the group action~\cite{DoranLasenby2003}. The stabilizer of \(R\) under conjugation is its centralizer,
$
Z_{\mathcal G}(R)
=
\{U\in\mathcal G:UR=RU\},
$
so its conjugacy orbit has the homogeneous-space description
\[
\mathcal O_R
\simeq
\mathcal G/Z_{\mathcal G}(R).
\]
When \(\mathcal G\) is a Lie subgroup of
\(\operatorname{Spin}(V,Q)\), its tangent space at \(R\) is
$
T_R\mathcal O_R
=
\{[X,R]:X\in\mathfrak g\},
$
while the elements satisfying
$
[X,R]=0
$
form the Lie algebra of the centralizer~\cite{Lee2012SmoothManifolds}.
For finite \(\mathcal G\), these orbits are the discrete conjugacy classes \(C_R\). Passing from individual group elements to conjugacy-class coordinates discards position within an orbit while retaining its class label. The discrete QCT converts these class coordinates into irreducible-character sectors,
\begin{equation}
\ket{C_R}
\overset{U_{\mathrm{QCT}}}{\longmapsto}
\sum_{\Gamma\in\widehat{\mathcal G}}
\widehat a_R(\Gamma)\ket{\Gamma}.
\end{equation}
This orbit-space formulation also motivates a continuous extension of the character transform, replacing the discrete conjugacy-class basis by functions on the conjugacy-orbit space of a compact group and the finite character table by the continuous character transform.

\section{Conclusion}

The quantum character transform is the discrete Fourier transform over the center \(Z(\mathbb{C}[\mathcal G])\) of a group algebra, which confirms its operator-theoretic setting. Conjugacy-class states and irreducible characters form dual bases of the same commutative algebra, and the QCT simultaneously diagonalizes every class-convolution operator. Central operators share a common spectral calculus in which convolution, Hamiltonian evolution, spectral filtering, and other matrix functions are diagonal operations on irreducible-character labels.

I further corroborate that the complexity of the QCT is closely tied to that of the non-abelian group Fourier transform, through its canonical restriction to the center. For finite groups with an efficient QFT, this gives the asymptotic QCT bound
\[
\mathcal O\!\left(
\operatorname{cost}(\mathrm{QFT}_{\mathcal G})
+
\operatorname{poly}\!\left(\log |\mathcal G|,\log(1/\epsilon)\right)
\right).
\]
Whenever the restriction overhead is asymptotically bounded by the QFT cost, this simplifies to $\operatorname{cost}(\mathrm{QCT}_{\mathcal G})
=
\mathcal O\!\left(
\operatorname{cost}(\mathrm{QFT}_{\mathcal G})
\right)$.
The transform applies whenever the relevant operator is conjugacy invariant. Normal Cayley graphs and their quantum walks, class-symmetric many-body Hamiltonians, exchange and molecular permutation dynamics, and symmetry-sector filters are all instances of the same central spectral problem. Still, the scope of the character transform should be stated precisely. Because class-convolution operators are central, they are block-scalar on irreducible sectors but do not retain the internal matrix-index information present in the full non-abelian Fourier transform. For this reason, the QCT should not be oversold as a universal replacement for non-abelian Fourier sampling in hidden subgroup problems. Over groups such as \(S_n\), there are strong limitations on what Fourier-sampling approaches can learn from single coset states, and those limitations should temper any claim that central class-operator measurements alone can overcome the known barriers surrounding graph-isomorphism-type instances. The correct message is narrower and stronger; central operators provide the right tool for conjugacy-invariant spectral questions, while non-central operators remain necessary for finer, less understood measurements beyond the class algebra.

The restriction method leaves a central complexity question unresolved. Most importantly, it establishes an upper bound inherited from the quantum Fourier transform but does not settle the optimal complexity of generating character tables on a quantum computer. Subsequent work demonstrates that QCT algorithms can improve upon this inherited bound without passing through the full group Fourier transform~\cite{RathoreZhou2026}. For groups such as \(S_n\), the gap between the group and class register sizes means that character transform algorithms specialized to the group and operating entirely in the compressed class space could use substantially lower resources than restriction-based algorithms. Discovering if more sub-Fourier implementations exist across different groups, how efficiently their character data can be computed, and whether these reductions translate into end-to-end advantages for symmetry-adapted simulation is the foreseeable future for the character transform.

\noindent
\newline
\textbf{Acknowledgements}.
This research was supported by the TRIO Ronald E. McNair Postbaccalaureate Achievement Program at North Carolina State University in summer 2025, which was funded by the U.S. Department of Education under grant number P217A220189-24.

\bibliographystyle{plain}
\bibliography{references}

@article{bastidas2025,
author = {Bastidas, Victor and Fitzpatrick, Nathan and Joven, K. and Rossi, Zane and Islam, Shariful and Voorhis, Troy and Chuang, Isaac and Liu, Yuan},
year = {2025},
month = {05},
pages = {},
title = {Unification of finite symmetries in the simulation of many-body systems on quantum computers},
volume = {111},
journal = {Physical Review A},
doi = {10.1103/PhysRevA.111.052433}
}

@article{frobenius1896gruppencharaktere,
  author  = {Frobenius, Ferdinand Georg},
  title   = {\"{U}ber die Gruppencharaktere},
  journal = {Sitzungsberichte der K\"{o}niglich Preu{\ss}ischen Akademie der Wissenschaften zu Berlin},
  pages   = {985--1021},
  year    = {1896}
}

@article{frobenius1897gruppendeterminante,
  author  = {Frobenius, Ferdinand Georg},
  title   = {\"{U}ber die Darstellung der endlichen Gruppen durch lineare Substitutionen},
  journal = {Sitzungsberichte der K\"{o}niglich Preu{\ss}ischen Akademie der Wissenschaften zu Berlin},
  pages   = {994--1015},
  year    = {1897}
}

@misc{RathoreZhou2026,
      title={An efficient Hamiltonian-based quantum algorithm for characters of the symmetric group}, 
      author={Dikshant Rathore and Leo Zhou},
      year={2026},
      eprint={2610.05752},
      archivePrefix={arXiv},
      primaryClass={quant-ph},
      url={https://arxiv.org/abs/2610.05752}, 
}

@book{burnside1897theory,
  author    = {Burnside, William},
  title     = {Theory of Groups of Finite Order},
  edition   = {1st},
  publisher = {Cambridge University Press},
  address   = {Cambridge},
  year      = {1897}
}

@article{schur1904darstellung,
  author  = {Schur, Issai},
  title   = {\"{U}ber die Darstellung der endlichen Gruppen durch abgeleitete lineare Substitutionen},
  journal = {Journal f\"{u}r die reine und angewandte Mathematik},
  volume  = {127},
  pages   = {20--50},
  year    = {1904}
}

@article{schur1905neue,
  author  = {Schur, Issai},
  title   = {Neue Begr\"{u}ndung der Theorie der Gruppencharaktere},
  journal = {Sitzungsberichte der K\"{o}niglich Preu{\ss}ischen Akademie der Wissenschaften zu Berlin},
  pages   = {406--432},
  year    = {1905}
}

@book{burnside1911theory,
  author    = {Burnside, William},
  title     = {Theory of Groups of Finite Order},
  edition   = {2nd},
  publisher = {Cambridge University Press},
  address   = {Cambridge},
  year      = {1911}
}

@book{hamermesh1962group,
  title={Group Theory and Its Application to Physical Problems},
  author={Hamermesh, Morton},
  year={1962},
  publisher={Addison-Wesley}
}

@book{birman1974theory,
  title={Theory of Crystal Space Groups and Infra-Red and Raman Lattice Processes of Insulating Crystals},
  author={Birman, Joseph L.},
  year={1974},
  publisher={Springer-Verlag}
}

@book{Bunker1979,
  author    = {Philip R. Bunker},
  title     = {Molecular Symmetry and Spectroscopy},
  publisher = {Academic Press},
  address   = {New York},
  year      = {1979},
  isbn      = {978-0-12-141350-7},
  doi       = {10.1016/B978-0-12-141350-7.X5001-0}
}

@inproceedings{Beals1997,
  author    = {Robert Beals},
  title     = {Quantum Computation of Fourier Transforms over Symmetric Groups},
  booktitle = {Proceedings of the Twenty-Ninth Annual ACM Symposium on Theory of Computing},
  series    = {STOC '97},
  pages     = {48--53},
  publisher = {Association for Computing Machinery},
  year      = {1997},
  doi       = {10.1145/258533.258548}
}

@article{BaconChuangHarrow2006,
  author  = {Dave Bacon and Isaac L. Chuang and Aram W. Harrow},
  title   = {Efficient Quantum Circuits for Schur and Clebsch-Gordan Transforms},
  journal = {Physical Review Letters},
  volume  = {97},
  number  = {17},
  pages   = {170502},
  year    = {2006},
  doi     = {10.1103/PhysRevLett.97.170502}
}

@article{Kuperberg2005,
  author  = {Greg Kuperberg},
  title   = {A Subexponential-Time Quantum Algorithm for the Dihedral Hidden Subgroup Problem},
  journal = {SIAM Journal on Computing},
  volume  = {35},
  number  = {1},
  pages   = {170--188},
  year    = {2005},
  doi     = {10.1137/S0097539703436345}
}

@article{jamiolkowski1972linear,
  title={Linear transformations which preserve trace and positive semidefiniteness of operators},
  author={Jamio{\l}kowski, Andrzej},
  journal={Reports on Mathematical Physics},
  volume={3},
  number={4},
  pages={275--278},
  year={1972},
  publisher={Elsevier},
  doi={10.1016/0034-4877(72)90011-0}
}

@article{choi1975completely,
  title={Completely positive linear maps on complex matrices},
  author={Choi, Man-Duen},
  journal={Linear Algebra and its Applications},
  volume={10},
  number={3},
  pages={285--290},
  year={1975},
  publisher={Elsevier},
  doi={10.1016/0024-3795(75)90075-0}
}

@book{cotton1990chemical,
  author    = {Cotton, F. Albert},
  title     = {Chemical Applications of Group Theory},
  edition   = {3rd},
  year      = {1990},
  publisher = {Wiley-Interscience},
  address   = {New York},
  isbn      = {978-0471510949}
}

@book{Diaconis1988,
  author    = {Diaconis, Persi},
  title     = {Group Representations in Probability and Statistics},
  publisher = {Institute of Mathematical Statistics},
  year      = {1988}
}

@book{Tinkham2003,
  author    = {Tinkham, Michael},
  title     = {Group Theory and Quantum Mechanics},
  publisher = {Dover},
  year      = {2003}
}

@book{BunkerJensen1998,
  author    = {Bunker, Philip R. and Jensen, Per},
  title     = {Molecular Symmetry and Spectroscopy},
  edition   = {2nd},
  publisher = {NRC Research Press},
  year      = {1998}
}

@book{BrouwerCohenNeumaier1989,
  author    = {Brouwer, Andries E. and Cohen, Arjeh M. and Neumaier, Arnold},
  title     = {Distance-Regular Graphs},
  series    = {Ergebnisse der Mathematik und ihrer Grenzgebiete. 3. Folge},
  volume    = {18},
  publisher = {Springer-Verlag},
  address   = {Berlin, Heidelberg},
  year      = {1989},
  isbn      = {978-3-540-50619-5}
}

@book{Godsil1993,
  author    = {Godsil, Christopher D.},
  title     = {Algebraic Combinatorics},
  series    = {Chapman \& Hall/CRC Mathematics Series},
  publisher = {Chapman \& Hall},
  address   = {New York},
  year      = {1993},
  isbn      = {978-0-412-04131-0}
}

@article{ArnadottirGodsil2023,
  author    = {\'{A}rnad\'{o}ttir, Arnbj\"{o}rg Soff\'{i}a and Godsil, Chris},
  title     = {A Note on Eigenvalues of {C}ayley Graphs},
  journal   = {Linear Algebra and its Applications},
  volume    = {721},
  pages     = {240--249},
  year      = {2025},
  doi       = {10.1016/j.laa.2023.12.002}
}

@article{Bravyi2025Characters,
  title = {Classical and Quantum Algorithms for Characters of the Symmetric Group},
  author = {Bravyi, Sergey and Gosset, David and Havlicek, Vojtech and Schatzki, Louis},
  journal = {PRX Quantum},
  volume = {6},
  issue = {3},
  pages = {030323},
  numpages = {15},
  year = {2025},
  month = {Aug},
  publisher = {American Physical Society},
  doi = {10.1103/bq28-r2r7},
  url = {https://link.aps.org/doi/10.1103/bq28-r2r7}
}

@article{EttingerHoyerKnill2004,
  author  = {Mark Ettinger and Peter H{\o}yer and Emanuel Knill},
  title   = {The Quantum Query Complexity of the Hidden Subgroup Problem Is Polynomial},
  journal = {Information Processing Letters},
  volume  = {91},
  number  = {1},
  pages   = {43--48},
  year    = {2004},
  doi     = {10.1016/j.ipl.2004.01.024}
}

@article{MooreRussellSchulman2008,
  author  = {Cristopher Moore and Alexander Russell and Leonard J. Schulman},
  title   = {The Symmetric Group Defies Strong Fourier Sampling},
  journal = {SIAM Journal on Computing},
  volume  = {37},
  number  = {6},
  pages   = {1842--1864},
  year    = {2008},
  doi     = {10.1137/050644896}
}

@article{Hoyer2002,
  author  = {Peter H{\o}yer},
  title   = {The Hidden Subgroup Problem and Quantum Computation Using Group Representations},
  journal = {SIAM Journal on Computing},
  volume  = {31},
  number  = {6},
  pages   = {1775--1787},
  year    = {2002},
  doi     = {10.1137/S009753970139450X}
}

@book{Serre1977,
  author    = {Jean-Pierre Serre},
  title     = {Linear Representations of Finite Groups},
  series    = {Graduate Texts in Mathematics},
  volume    = {42},
  publisher = {Springer},
  year      = {1977},
  doi       = {10.1007/978-1-4684-9458-7}
}

@article{Sunada1985,
  author  = {Toshikazu Sunada},
  title   = {Riemannian Coverings and Isospectral Manifolds},
  journal = {Annals of Mathematics},
  volume  = {121},
  number  = {1},
  pages   = {169--186},
  year    = {1985},
  doi     = {10.2307/1971195}
}

@article{Perlis1977,
  author  = {Robert Perlis},
  title   = {On the Equation \(\zeta_K(s)=\zeta_{K'}(s)\)},
  journal = {Journal of Number Theory},
  volume  = {9},
  number  = {3},
  pages   = {342--360},
  year    = {1977},
  doi     = {10.1016/0022-314X(77)90070-1}
}

@unpublished{BermejoVegaZatloukal2015,
  author        = {Juan Bermejo-Vega and Kevin C. Zatloukal},
  title         = {Abelian Hypergroups and Quantum Computation},
  year          = {2015},
  note          = {arXiv:1509.05806 [quant-ph]},
  eprint        = {1509.05806},
  archivePrefix = {arXiv},
  primaryClass  = {quant-ph},
  url           = {https://arxiv.org/abs/1509.05806}
}

@article{Castelazo2022GroupConvolution,
  author  = {Grecia Castelazo and Quynh T. Nguyen and Giacomo De Palma
             and Dirk Englund and Seth Lloyd and Bobak T. Kiani},
  title   = {Quantum Algorithms for Group Convolution, Cross-Correlation,
             and Equivariant Transformations},
  journal = {Physical Review A},
  volume  = {106},
  number  = {3},
  pages   = {032402},
  year    = {2022},
  month   = {September},
  doi     = {10.1103/PhysRevA.106.032402}
}

@inproceedings{Gilyen2019QSVT,
  author    = {Andr{\'a}s Gily{\'e}n and Yuan Su and Guang Hao Low
               and Nathan Wiebe},
  title     = {Quantum Singular Value Transformation and Beyond:
               Exponential Improvements for Quantum Matrix Arithmetics},
  booktitle = {Proceedings of the 51st Annual ACM SIGACT Symposium
               on Theory of Computing},
  pages     = {193--204},
  year      = {2019},
  doi       = {10.1145/3313276.3316366}
}

@article{mcardle2020quantum,
  title   = {Quantum Computational Chemistry},
  author  = {McArdle, Sam and Endo, Suguru and Aspuru-Guzik, Al{\'a}n and Benjamin, Simon C. and Yuan, Xiao},
  journal = {Reviews of Modern Physics},
  volume  = {92},
  number  = {1},
  pages   = {015003},
  year    = {2020},
  doi     = {10.1103/RevModPhys.92.015003}
}

@article{cao2019quantum,
  title   = {Quantum Chemistry in the Age of Quantum Computing},
  author  = {Cao, Yudong and Romero, Jonathan and Olson, Jonathan P. and Degroote, Matthias and Johnson, Peter D. and Kieferov{\'a}, M{\'a}ria and Kivlichan, Ian D. and Menke, Tim and Peropadre, Borja and Sawaya, Nicolas P. D. and Sim, Sukin and Veis, Libor and Aspuru-Guzik, Al{\'a}n},
  journal = {Chemical Reviews},
  volume  = {119},
  number  = {19},
  pages   = {10856--10915},
  year    = {2019},
  doi     = {10.1021/acs.chemrev.8b00803}
}

@article{Bacon2008,
  author  = {Bacon, Dave},
  title   = {How a Clebsch--Gordan Transform Helps to Solve the
             Heisenberg Hidden Subgroup Problem},
  journal = {Quantum Information and Computation},
  volume  = {8},
  number  = {5},
  pages   = {438--467},
  year    = {2008},
  doi     = {10.26421/QIC8.5-6},
  eprint  = {quant-ph/0612107},
  archivePrefix = {arXiv}
}

@article{HallgrenEtAl2010,
  author  = {Hallgren, Sean and Moore, Cristopher and R{\"o}tteler, Martin
             and Russell, Alexander and Sen, Pranab},
  title   = {Limitations of Quantum Coset States for Graph Isomorphism},
  journal = {Journal of the ACM},
  volume  = {57},
  number  = {6},
  articleno = {34},
  year    = {2010},
  doi     = {10.1145/1857914.1857918}
}

@article{MooreRussellSniady2010,
  author  = {Moore, Cristopher and Russell, Alexander and {\'S}niady, Piotr},
  title   = {On the Impossibility of a Quantum Sieve Algorithm for Graph Isomorphism},
  journal = {SIAM Journal on Computing},
  volume  = {39},
  number  = {6},
  pages   = {2377--2396},
  year    = {2010},
  doi     = {10.1137/080724101}
}

@book{Terras1999,
  author    = {Terras, Audrey},
  title     = {Fourier Analysis on Finite Groups and Applications},
  publisher = {Cambridge University Press},
  year      = {1999},
  series    = {London Mathematical Society Student Texts},
  volume    = {43},
  doi       = {10.1017/CBO9780511626265}
}

@article{Babai1979,
  author  = {Babai, L{\'a}szl{\'o}},
  title   = {Spectra of Cayley Graphs},
  journal = {Journal of Combinatorial Theory, Series B},
  volume  = {27},
  number  = {2},
  pages   = {180--189},
  year    = {1979},
  doi     = {10.1016/0095-8956(79)90079-0}
}

@misc{gromada2026quantumassociationschemes,
      title={Quantum association schemes}, 
      author={Daniel Gromada},
      year={2024},
      eprint={2404.06157},
      archivePrefix={arXiv},
      primaryClass={math.QA},
      url={https://arxiv.org/abs/2404.06157}, 
}

@inproceedings{CohenWelling2016,
  author    = {Cohen, Taco S. and Welling, Max},
  title     = {Group Equivariant Convolutional Networks},
  booktitle = {Proceedings of the 33rd International Conference on Machine Learning},
  series    = {Proceedings of Machine Learning Research},
  volume    = {48},
  pages     = {2990--2999},
  year      = {2016}
}

@article{BronsteinEtAl2021,
  author  = {Bronstein, Michael M. and Bruna, Joan and Cohen, Taco
             and Veli{\v{c}}kovi{\'c}, Petar},
  title   = {Geometric Deep Learning: Grids, Groups, Graphs, Geodesics,
             and Gauges},
  journal = {arXiv preprint arXiv:2104.13478},
  year    = {2021}
}

@inproceedings{DefferrardEtAl2016,
  author    = {Defferrard, Micha{\"e}l and Bresson, Xavier
               and Vandergheynst, Pierre},
  title     = {Convolutional Neural Networks on Graphs with Fast Localized
               Spectral Filtering},
  booktitle = {Advances in Neural Information Processing Systems},
  volume    = {29},
  year      = {2016}
}

@inproceedings{KipfWelling2017,
  author    = {Kipf, Thomas N. and Welling, Max},
  title     = {Semi-Supervised Classification with Graph Convolutional Networks},
  booktitle = {International Conference on Learning Representations},
  year      = {2017}
}

@article{HarrisonScott2003,
  author  = {Harrison, P. F. and Scott, W. G.},
  title   = {Permutation Symmetry, Tri-Bimaximal Neutrino Mixing and the
             \(S_3\) Group Characters},
  journal = {Physics Letters B},
  volume  = {557},
  number  = {1--2},
  pages   = {76--86},
  year    = {2003},
  doi     = {10.1016/S0370-2693(03)00183-7}
}

@article{clausen1989fast,
  author    = {Clausen, Michael},
  title     = {Fast {Fourier} Transforms for Metabelian Groups},
  journal   = {SIAM Journal on Computing},
  volume    = {18},
  number    = {3},
  pages     = {584--593},
  year      = {1989},
  doi       = {10.1137/0218040},
  url       = {https://doi.org/10.1137/0218040}
}

@misc{hoyer1997efficient,
  author        = {Hoyer, Peter},
  title         = {Efficient Quantum Transforms},
  year          = {1997},
  eprint        = {quant-ph/9702028},
  archivePrefix = {arXiv},
  primaryClass  = {quant-ph}
}

@article{KawanoSekigawa2016,
  author  = {Yasuhito Kawano and Hiroshi Sekigawa},
  title   = {Quantum Fourier Transform over Symmetric Groups---Improved Result},
  journal = {Journal of Symbolic Computation},
  volume  = {75},
  pages   = {219--243},
  year    = {2016},
  doi     = {10.1016/j.jsc.2015.11.016}
}

@article{BruinsmaGrinkoOzols2026,
  author        = {Carli Bruinsma and Dmitry Grinko and Maris Ozols},
  title         = {Quantum Fourier Transform for the Symmetric Group},
  journal       = {arXiv preprint arXiv:2608.28569},
  year          = {2026},
  eprint        = {2608.28569},
  archivePrefix = {arXiv},
  primaryClass  = {quant-ph}
}

@article{Katriel1998ClassSum,
  author  = {Jacob Katriel},
  title   = {Class-Sum Products in the Symmetric Group: Combinatorial Interpretation of the Reduced Class Coefficients},
  journal = {International Journal of Quantum Chemistry},
  volume  = {68},
  number  = {2},
  pages   = {103--118},
  year    = {1998}
}

@article{Gassmann1926,
  author  = {Gassmann, Fritz},
  title   = {Bemerkungen zur vorstehenden Arbeit von Hurwitz:
             {\"U}ber Beziehungen zwischen den Primidealen eines algebraischen
             K{\"o}rpers und den Substitutionen seiner Gruppe},
  journal = {Mathematische Zeitschrift},
  volume  = {25},
  pages   = {665--675},
  year    = {1926}
}

@article{ShumanEtAl2013,
  author  = {Shuman, David I. and Narang, Sunil K. and Frossard, Pascal
             and Ortega, Antonio and Vandergheynst, Pierre},
  title   = {The Emerging Field of Signal Processing on Graphs:
             Extending High-Dimensional Data Analysis to Networks
             and Other Irregular Domains},
  journal = {IEEE Signal Processing Magazine},
  volume  = {30},
  number  = {3},
  pages   = {83--98},
  year    = {2013},
  doi     = {10.1109/MSP.2012.2235192}
}

@article{HammondEtAl2011,
  author  = {Hammond, David K. and Vandergheynst, Pierre
             and Gribonval, R{\'e}mi},
  title   = {Wavelets on Graphs via Spectral Graph Theory},
  journal = {Applied and Computational Harmonic Analysis},
  volume  = {30},
  number  = {2},
  pages   = {129--150},
  year    = {2011},
  doi     = {10.1016/j.acha.2010.04.005}
}

@article{FarhiGutmann1998,
  author  = {Farhi, Edward and Gutmann, Sam},
  title   = {Quantum Computation and Decision Trees},
  journal = {Physical Review A},
  volume  = {58},
  number  = {2},
  pages   = {915--928},
  year    = {1998},
  doi     = {10.1103/PhysRevA.58.915}
}

@book{DoranLasenby2003,
  author    = {Doran, Chris and Lasenby, Anthony},
  title     = {Geometric Algebra for Physicists},
  publisher = {Cambridge University Press},
  year      = {2003},
  doi       = {10.1017/CBO9780511807497}
}

@book{Lee2012SmoothManifolds,
  author    = {Lee, John M.},
  title     = {Introduction to Smooth Manifolds},
  edition   = {2},
  series    = {Graduate Texts in Mathematics},
  volume    = {218},
  publisher = {Springer},
  year      = {2012},
  doi       = {10.1007/978-1-4419-9982-5}
}

@book{FultonHarris1991,
  author    = {Fulton, William and Harris, Joe},
  title     = {Representation Theory: A First Course},
  series    = {Graduate Texts in Mathematics},
  volume    = {129},
  publisher = {Springer-Verlag},
  address   = {New York},
  year      = {1991},
  doi       = {10.1007/978-1-4612-0979-9},
  isbn      = {978-0-387-97495-8}
}

@article{BeckEtAl2024,
  author  = {Beck, Kathryn and Ghandehari, Mahya and Hudson, Skyler
             and Paltenstein, Jenna},
  title   = {Frames for Signal Processing on Cayley Graphs},
  journal = {Journal of Fourier Analysis and Applications},
  volume  = {30},
  number  = {6},
  pages   = {66},
  year    = {2024},
  doi     = {10.1007/s00041-024-10128-5},
  eprint  = {2303.02812},
  archivePrefix = {arXiv},
  primaryClass  = {cs.IT}
}

\end{document}